\documentclass[11pt]{article}
\usepackage{fullpage}
\usepackage[T1]{fontenc}
\usepackage[utf8]{inputenc}
\usepackage{xcolor}
\usepackage[colorlinks=true]{hyperref}
\definecolor{purple}{HTML}{7068A4}
\definecolor{blue}{HTML}{3F70B2}
\definecolor{orange}{HTML}{A47458}
\definecolor{yellow}{HTML}{9B7A3C}
\hypersetup{
    linkcolor=blue
    ,citecolor=orange
    ,filecolor=purple
    ,urlcolor=orange
    ,menucolor=purple
    ,runcolor=purple
}
\usepackage{amsmath, amssymb, amsthm, amsfonts, graphicx, color, subcaption, enumerate, bm, cleveref, enumitem, array, mathtools}
\usepackage{enumitem}
 \usepackage[colorinlistoftodos,prependcaption,textsize=tiny]{todonotes}
\usepackage{multicol, multirow}
\usepackage{thmtools, thm-restate}
\usepackage{ragged2e}
\usepackage{lineno}
\usepackage{framed}
\usepackage[framemethod=tikz]{mdframed}
\usepackage{enumitem}

\usepackage[ruled,linesnumbered]{algorithm2e}
\SetKwComment{Comment}{/* }{ */}

\usepackage{subfiles}
\graphicspath{{graphics/}}

\crefname{claim}{Claim}{Claims}
\crefname{property}{Property}{Properties}
\crefname{algocf}{Algorithm}{Algorithms}
\Crefname{algocf}{Algorithm}{Algorithms}
\crefname{step}{Step}{Steps}
\Crefname{step}{Step}{Steps}
\usepackage{tcolorbox}
\usepackage[style=trad-alpha,natbib=true,maxcitenames=4]{biblatex}
\newtheorem{theorem}{Theorem}
\newtheorem{lemma}{Lemma}[section]

\newtheorem{corollary}[lemma]{Corollary}

\newtheorem{proposition}[lemma]{Proposition}
\newtheorem{problem}[theorem]{Problem}

\theoremstyle{definition}

\theoremstyle{remark}
\newtheorem{remark}[lemma]{Remark}
\newtheorem*{remark*}{Remark}

\newcommand{\LOCAL}{$\mathsf{LOCAL}$}
\newcommand{\CONGEST}{$\mathsf{CONGEST}$}
\newcommand{\CCLIQUE}{$\mathsf{CONGESTED~ CLIQUE}$}

\renewcommand{\epsilon}{\varepsilon}
\newcommand{\poly}{\operatorname{poly}}

\newcommand{\E}{\mathbb{E}}

\renewcommand{\emptyset}{\varnothing}
\newcommand{\randcolortrial}{\textbf{$\mathsf{Random\mbox{-}Color\mbox{-}Trial}$}\xspace}
\newcommand{\colorsamp}{\textbf{$\mathsf{Color\mbox{-}Sample}$}\xspace}

\newcommand{\DLC}{\textsf{D1LC}\xspace}
\newcommand{\zec}{\textsf{ZEC}\xspace}
\newcommand{\zecnew}{\textsf{ZEC-NEW}\xspace}
\newcommand{\bea}{\mathbf{E}_A\xspace}
\newcommand{\beb}{\mathbf{E}_B\xspace}
\newcommand{\labela}{ \mathsf{L}_A\xspace}
\newcommand{\labelb}{ \mathsf{L}_B\xspace}
\newcommand{\comments}[1]{\textcolor{black}{#1}}

\title{Round and Communication Efficient Graph Coloring}
 \author{Yi-Jun Chang\footnote{National University of Singapore. ORCID: 0000-0002-0109-2432. Email: cyijun@nus.edu.sg} \and Gopinath Mishra\footnote{National University of Singapore. ORCID: 0000-0003-0540-0292. Email: gopinath@nus.edu.sg} \and Hung Thuan Nguyen\footnote{National University of Singapore. ORCID: 0009-0006-7993-2952.  Email: hung@u.nus.edu} \and Farrel D Salim\footnote{National University of Singapore. ORCID: 0009-0000-8969-0279. Email: farrel\_salim@u.nus.edu} }

\date{}

\begin{document}

\maketitle
\thispagestyle{empty}

\begin{abstract}
In the context of communication complexity, we explore protocols for graph coloring, focusing on the vertex and edge coloring problems in $n$-vertex graphs $G$ with a maximum degree $\Delta$. We consider a scenario where the edges of $G$ are partitioned between two players. 

Our first contribution is a randomized protocol that efficiently finds a $(\Delta + 1)$-vertex coloring of $G$, utilizing $O(n)$ bits of communication in expectation and completing in $O(\log \log n \cdot \log \Delta)$ rounds in the worst case.  This advancement represents a significant improvement over the work of Flin and Mittal [\emph{Distributed Computing} 2025], who achieved the same communication cost but required $O(n)$ rounds in expectation, thereby making a significant reduction in the round complexity. 

Our second contribution is a  {deterministic} protocol to compute a $(2\Delta - 1)$-edge coloring of $G$, which maintains the same $O(n)$ bits of communication and uses only $O(1)$ rounds. % in the worst case.  
We complement the result with a tight $\Omega(n)$-bit lower bound on the communication complexity of the $(2\Delta-1)$-edge coloring problem, while a similar $\Omega(n)$ lower bound for the $(\Delta+1)$-vertex coloring problem has been established by Flin and Mittal [\emph{Distributed Computing} 2025]. Our result implies a space lower bound of $\Omega(n)$ bits for $(2\Delta - 1)$-edge coloring in the $W$-streaming model, which is the first non-trivial space lower bound for edge coloring in the $W$-streaming model. 

\end{abstract}

%\listoftodos

\newpage
\bigskip
\tableofcontents
\bigskip
\thispagestyle{empty}

% \newpage

% \listoftodos

\newpage
\pagenumbering{arabic}
\section{Introduction} 

The \emph{two-party communication model} was introduced by Yao \cite{yao1979complexity} and serves as a foundational framework in various fields, including theoretical computer science. In this model (for graphs), two parties, Alice and Bob, each hold part of the input—specifically, portions of a graph where the edges are partitioned between them. They must communicate to compute a function of their combined inputs, such as finding a valid $(\Delta + 1)$-vertex coloring or $(2\Delta - 1)$-edge coloring. The focus is on minimizing the total number of bits exchanged, known as \emph{communication complexity}, which is defined as the minimum number of bits required for Alice and Bob to compute the desired function. Additionally, \emph{round complexity} refers to the number of communication rounds needed for the two parties to reach a solution, where each round consists of messages exchanged between them. Efficient protocols aim to minimize both communication and round complexity, which is especially valuable in large-scale distributed systems where bandwidth is limited and rapid decision-making is essential. For a thorough exposition of the two-party communication model, we recommend referring to \cite{kushilevitz1997communication,rao2020communication}. This model allows us to investigate the communication costs associated with solving the aforementioned coloring problems, highlighting trade-offs between communication complexity and round efficiency, and leading to new insights for distributed protocols. 

{The interplay between round and communication complexity was initiated by Papadimitriou and Sipser \cite{PapadimitriouSipser1982}, who demonstrated that bounding rounds can substantially increase communication. This foundational perspective has shaped decades of work in communication complexity \cite{kushilevitz1997communication,rao2020communication}.  In this work, we present graph coloring protocols that are round-efficient and communication-optimal, meaning that our protocols achieve optimal communication complexity while using \emph{significantly} fewer rounds.}
% In this work, we focus on graph coloring problems within the communication complexity model.

% Graph coloring is a fundamental problem in computer science, with applications ranging from scheduling and resource allocation to frequency assignment in wireless networks. It involves the assignment of colors to elements of a graph, such as its vertices or edges, subject to specific constraints. In vertex coloring, the goal is to assign colors to vertices such that no two adjacent vertices share the same color, while in edge coloring, the objective is to color edges so that no two edges incident on the same vertex receive the same color. A well-known result is that any graph with a maximum degree $\Delta$ can be greedily colored with $\Delta + 1$ colors for a vertex coloring. Similarly, for a $(2\Delta - 1)$-edge coloring, greedy algorithms can produce a proper coloring, although this bound can often be suboptimal in practice. Vizing’s theorem guarantees that a $(\Delta + 1)$-edge coloring is achievable \cite{vizing1964}. 

Graph coloring is a fundamental problem in computer science, with a wide range of applications, including scheduling, resource allocation, and frequency assignment in wireless networks. The problem involves assigning colors to the elements of a graph, such as its vertices or edges, subject to specific constraints. In \emph{vertex coloring}, the goal is to assign colors to the vertices of a graph such that no two adjacent vertices share the same color, while in \emph{edge coloring}, the objective is to color the edges so that no two edges incident to the same vertex receive the same color. A well-known result in vertex coloring is that any graph with maximum degree $\Delta$ can be greedily colored using at most $\Delta + 1$ colors. Similarly, for edge coloring, a greedy algorithm can achieve a proper coloring using $2\Delta - 1$ colors. However, this bound is often suboptimal in practice. Vizing's theorem guarantees that a proper edge coloring can be achieved using at most $\Delta + 1$ colors for any graph \cite{vizing1964}.

Substantial research has been devoted to graph algorithms in the context of communication complexity~\cite{kushilevitz1997communication, rao2020communication, hajnal1988communication, assadi2021simple, babai1986complexity, razborov1992distributional, ivanyos2012new}, in particular the \emph{two-party edge-partition} model~\cite{assadi2024log, assadi2023noisy, assadi2023coloring, blikstad2022communication}, where the edge set is arbitrarily divided between two players.  However, the study of coloring problems within the communication complexity framework has received comparatively little attention. In contrast, there exists a substantial body of work on $(\Delta + 1)$-vertex coloring and $(2\Delta - 1)$-edge coloring in various \emph{restricted} computational models, such as streaming, dynamic, parallel, and distributed settings~\cite{assadi2019sublinear, linial1992locality, barenboim2016locality, chang2020distributed, ghaffari2021deterministic, chang2019complexity, czumaj2021deterministic, bhattacharya2018dynamic, DBLP:journals/talg/BhattacharyaGKL22, DBLP:journals/talg/HenzingerP22, DBLP:journals/siamcomp/Harris20, DBLP:journals/tcs/HalldorssonN23, bhattacharya2018dynamic}.

\paragraph{Coloring in the two-party edge-partition communication model:}
The $(\Delta+1)$-vertex coloring problem has been recently studied in the two-party communication model. As argued in~\cite{assadi2023coloring}, it is not hard to obtain a \emph{deterministic} protocol for the $(\Delta+1)$-vertex coloring problem with communication complexity $O(n \poly(\log n))$ by simulating the greedy coloring algorithm and using a binary search for finding available colors. However, this approach requires a significant number of communication rounds. 
%A straightforward greedy coloring approach, explored in \cite{10.1145/3662158.3662796}, requires $O(n \log \Delta)$ bits of communication and rounds. This protocol iteratively assigns random colors to vertices based on a random permutation and communicates 2 bits per vertex to check color validity. 
\citeauthor{assadi2023coloring}~\cite[Corollary 3.11 of arXiv:2212.10641v1]{assadi2023coloring} presented a deterministic vertex coloring algorithm using $O(n \log^4 n)$ bits of communication and $O(\log \Delta \log \log \Delta)$ rounds, which is efficient in both round and communication complexities.

Recently, \citet{10.1145/3662158.3662796} presented a randomized vertex coloring algorithm using $O(n)$ bits of communication \emph{in expectation}. Moreover, they showed a lower bound $\Omega(n)$ on the communication complexity matching their upper bound. However, just like the greedy coloring algorithm, the algorithm of \citet{10.1145/3662158.3662796} is inherently sequential and not round-efficient. 

%An intriguing question remained open: Is it possible to design a \emph{round-efficient} protocol for $(\Delta+1)$-vertex coloring with an \emph{optimal} communication complexity $O(n)$?

% In this paper, we focus on optimizing both communication and round complexity for the $(\Delta+1)$-vertex coloring and $(2\Delta-1)$-edge coloring problems in the two-party communication model. For the $(\Delta+1)$-vertex coloring problem, we present a protocol with $O(n)$ expected communication bits and $O(\log \log n \log \Delta)$ rounds, significantly reducing the round complexity. To the best of our knowledge, we provide the first nontrivial protocol for the $(2\Delta-1)$-edge problem that uses $O(n)$ bits of communication in expectation. Note that in the $(\Delta+1)$-vertex coloring problem, both parties must output colorings for all vertices, while in the $(2\Delta-1)$-edge coloring problem, they only need to output colorings for their respective edges.

%\yijun{Maybe start with the trivial algo based on greedy coloring, and then explain that Assadi's algorithm is interesting because the round-communication tradeoff.}

\subsection{Our contribution and comparison with previous works}\label{sec:contribution}
In this paper, we design a new protocol for the $(\Delta+1)$-vertex coloring problem that is both \emph{round-efficient} and \emph{communication-optimal}. Specifically, we present a protocol with $O(n)$ bits of communication in expectation and $O(\log \log n \cdot \log \Delta)$ rounds in the worst case. Moreover, we obtain an analogous result for edge coloring. Specifically, we present a deterministic protocol for the $(2\Delta-1)$-edge coloring problem with $O(n)$ bits of communication and $O(1)$ rounds, and we show a matching lower bound of $\Omega(n)$ bits on the communication complexity. 
%and $(2\Delta-1)$-edge coloring problems in the two-party communication model. For $(\Delta+1)$-vertex coloring, we present a protocol with $O(n)$ expected communication bits and $O(\log \log n \log \Delta)$ rounds, significantly reducing the round complexity. 
To the best of our knowledge, this is the first nontrivial $(2\Delta-1)$-edge coloring communication protocol.  %using $O(n)$ expected communication bits. 

We emphasize that, for the $(\Delta+1)$-vertex coloring problem, both parties must output colorings for all vertices, whereas for the $(2\Delta-1)$-edge coloring problem, each party only needs to output colorings for its respective edges.

%In this paper, we focus on optimizing both communication and round complexity for the $(\Delta+1)$-vertex coloring and $(2\Delta-1)$-edge coloring problems in the two-party communication model. For $(\Delta+1)$-vertex coloring, we present a protocol with $O(n)$ expected communication bits and $O(\log \log n \log \Delta)$ rounds, significantly reducing the round complexity. To the best of our knowledge, we provide the first nontrivial $(2\Delta-1)$-edge coloring protocol using $O(n)$ expected communication bits. Note that in the $(\Delta+1)$-vertex coloring problem, both parties must output colorings for all vertices, while in the $(2\Delta-1)$-edge coloring problem, they only need to output colorings for their respective edges.

% In this paper, we introduce a new randomized algorithm for the $(\Delta + 1)$-coloring communication problem. While a lower bound of $\Omega(n)$ bits complexity has been established in previous work by \citet{10.1145/3662158.3662796}, the round complexity of such protocols has not been previously studied.\yijun{Discuss some prior work about the round complexity of communication protocols, so that it becomes clear that there is significant interest in optimizing the round complexity of communication protocols. There should be many such papers, see e.g., \cite{huang2021communication}.} We aim to address this gap by presenting a protocol that runs in $O(\log \log n \cdot \log \Delta)$ rounds in the worst case while maintaining an expected communication complexity of $O(n)$ bits.

\paragraph{Vertex coloring:} Our result for the $(\Delta+1)$-vertex coloring problem is stated in the following theorem:%, and that for the $(2\Delta-1)$-edge coloring problem is presented in the subsequent theorem.

\begin{restatable}{theorem}{mainvertex}\label{thm:vertex}
    There exists a randomized protocol that, given an $n$-vertex graph $G$ with maximum degree $\Delta$, finds a $(\Delta + 1)$-vertex coloring of $G$ using $O(n)$ bits of communication in expectation and completes in $O(\log \log n \cdot \log \Delta)$ rounds in the worst case.
\end{restatable}

We compare \Cref{thm:vertex} with previous works. 
Some existing streaming algorithms already implicitly imply communication protocols in the two-party communication model.
The $(\Delta+1)$-vertex coloring problem admits a one-round randomized communication protocol with a communication complexity of $O(n \log^3 n)$ bits, due to the one-pass streaming algorithm of \citet{assadi2019sublinear},\footnote{{The transformation from a one-pass streaming algorithm to a \emph{one-round} communication protocol relies on the assumption that, in one round, Alice and Bob can each send a message to the other simultaneously.}} while its lower bound stands at $\Omega(n)$ bits~\cite{10.1145/3662158.3662796}. On the other hand, the deterministic $O(\log \Delta \cdot \log \log \Delta)$-pass streaming algorithm using $O(n \log^2 n)$ bits of space for $(\Delta+1)$-vertex coloring by \citet{assadi2023coloring} implies a deterministic communication protocol with $O(n \log^2 n \cdot \log \Delta \cdot \log \log \Delta)$ communication cost. %\hung{I think it should be \(n \log^4 n\) as stated in their paper. This equals to round times space in streaming model.}\gopinath{Should it be $O\left(n \log^2 n \cdot \log \Delta \cdot \log \log \Delta \right)$ or $O(n \log^4 n)$?}
%\footnote{Communication lower bounds have played a crucial role in deriving a broad spectrum of impossibility results across various domains, including distributed computing \cite{sarma2011distributed,bachrach2019hardness}, streaming \cite{indyk2003tight,miltersen1998data}, property testing \cite{DBLP:journals/toct/BlaisCG19,DBLP:conf/approx/EdenR18}, data structures \cite{miltersen1998data}, and circuit complexity \cite{karchmer1990monotone}. Therefore, it is anticipated that improved communication protocols with enhanced bit complexity can be achieved compared to space or communication costs in other models, particularly in the streaming context.} 
\citet{10.1145/3662158.3662796} closed the gap between nearly linear upper bounds and the exact linear lower bound by designing a communication protocol with $O(n)$ bits of expected communication. Similar gaps between nearly linear upper bounds and exact linear lower bounds have been previously observed in the communication complexity for many problems, such as  \emph{connectivity, minimum cut, and maximum matching}. See \Cref{sec:related} for details. Though the algorithm of Flin and Mittal is communication-efficient, it spends $O(n)$ rounds in expectation, implying that the protocol is not round-efficient. Our result on the $(\Delta+1)$-vertex coloring problem in \Cref{thm:vertex} significantly improves the round complexity as compared to Flin and Mittal, while maintaining the same communication cost. That is, our protocol is both communication-optimal and round-efficient.%\yijun{Also compare our work with \citet{assadi2023coloring}}

\paragraph{Edge coloring:} Our second result is a {deterministic} communication protocol to find a $(2\Delta-1)$-edge coloring, as stated in the following theorem:

\begin{restatable}{theorem}{mainedgeupper}\label{thm:edge}
    There exists a {deterministic} protocol that, given an $n$-vertex graph $G$ with maximum degree $\Delta$, finds a $(2\Delta - 1)$-edge coloring of $G$ using $O(n)$ bits of communication  and completes in {$O(1)$} rounds. %\hung{Do we need "in the worst case here?} \yijun{I think no need for deterministic protocols}\gopinath{Yes!}
    %\farrel{The worst-case round complexity is wrong?}\yijun{Should be correct?}\hung{Yes, I think it should be correct.}
    %\yijun{Can you do $O(\log \log n)$ rounds and then now the number of remaining uncolored vertices are just $O(n / \log n)$ in expectation so we can brute force in one round? Actually... I am thinking even $O(\log^\ast \Delta)$ is doable by sampling more colors in subsequent rounds}
\end{restatable}
To the best of our knowledge, no non-trivial $(2\Delta - 1)$-edge coloring protocol with a communication cost of \(o(n^2)\) was known before our work. However, for \(O(\Delta)\)-edge coloring, there is a one-round randomized  communication protocol with a communication cost of \(\widetilde{O}(n \sqrt{\Delta})\), which follows from the corresponding one-pass $W$-streaming algorithm \cite{DBLP:conf/icalp/SaneianB24}. \comments{In the $W$-streaming model, the output is also allowed to be reported in a streaming fashion, so the space complexity can be smaller than the output size. For the edge coloring problem, the $W$-streaming model is a natural model since the output size is as large as the input, making it impractical to store the entire output and return it all at once.}

\comments{In contrast to \Cref{thm:edge}, we observe that \((2\Delta)\)-edge coloring can be solved deterministically using no communication.}

\begin{theorem}\label{thm:2delta}
    There exists a deterministic protocol for \((2\Delta)\)-edge coloring without communication.
\end{theorem}
%\gopinath{May be having this result in the introduction when discussing our contribution would be nice.}\yijun{I agree}

\begin{proof}
We use an extension of Vizing's theorem by \citet{fournier1973colorations} which states that if the vertices of maximum degree \(\Delta\) form an independent set, then the number of colors needed for edge coloring is \(\Delta\). Alice and Bob partition the set of colors in half, with each party having a palette of \(\Delta\) colors. Because the edges are partitioned between the two parties; vertices of degree \(\Delta\) in one party have no edges in the other party. Each party sequentially identifies edges connecting two vertices of degree \(\Delta\) and colors them using any color from the other party's palette. Consequently, the remaining subgraph for each party satisfies the condition that no vertices of degree \(\Delta\) are connected, allowing them to $\Delta$-edge color their respective subgraph utilizing their respective palettes of size $\Delta$.
%Initially, each party finds a maximum matching of vertices of degree \(\Delta\) in their respective graphs and colors the edges in this matching using \emph{any} color from the other party's palette. The remaining subgraphs of both parties satisfy the condition that no pair of vertices of degree \(\Delta\) is connected. Therefore, by \Cref{thm:independent}, a proper edge coloring with \(\Delta\) colors exists, allowing them to use their respective sets of colors to color their remaining edges.
 %This approach enables us to establish a simple zero-communication protocol for the \(2\Delta\)-edge coloring
\end{proof}

%Hence, \Cref{thm:edge} is the first (non-trivial) result within the framework of two-party communication. 

Our protocol for the $(2\Delta - 1)$-edge coloring problem is both round-efficient and communication-optimal. To demonstrate the optimality of the communication cost of \Cref{thm:edge}, we prove the following theorem.

\begin{restatable}{theorem}{edgecoloringlower}\label{them:lower} 
   Any randomized protocol that finds a $(2\Delta - 1)$-edge coloring of the input graph with probability $1/2$ requires $\Omega(n)$ bits of communication in the worst case.
\end{restatable}

The number of colors in the lower bound of \Cref{them:lower} is optimal in the following sense: If we use more than $2\Delta-1$ colors, then an edge coloring can be computed without any communication, due to \Cref{thm:2delta}.  

The lower bound of \Cref{them:lower} applies not only to deterministic algorithms but also to randomized algorithms with constant success probability. Furthermore, by Markov's inequality, it extends to \emph{Las Vegas} algorithms as well.

\begin{corollary}\label{cor-lb-main}
    Any protocol that solves the $(2\Delta - 1)$-edge coloring problem requires $\Omega(n)$ bits of communication in expectation.%\yijun{Can move this to the intro, as this is just a corollary with a short proof}
\end{corollary}
\begin{proof}
    If there exists a protocol that finds a $(2\Delta - 1)$-edge coloring of the input graph using $o(n)$ expected bits of communication, then we can construct another protocol that finds a $(2\Delta - 1)$-edge coloring using $o(n)$ bits of communication in the worst case, with a success probability of at least ${1}/{2}$ by Markov's inequality. This would contradict \Cref{them:lower}.
\end{proof}

The communication complexity lower bounds in the two-party communication model translate to a space lower bound in the standard streaming setting, where output size also counts toward space usage. Therefore, \Cref{cor-lb-main} implies an $\Omega(n)$-bit space lower bound for the $(2\Delta-1)$-edge coloring problem in the standard streaming setting. As already mentioned, in the context of edge coloring, the $W$-streaming model is more natural. However,  the reduction from the two-party communication model to the $W$-streaming model is not direct, specifically in the context of the $(2\Delta-1)$-edge coloring problem.

Let us consider a slight modification of the $(2\Delta-1)$-edge coloring problem in the two-party communication model. Instead of Alice and Bob being required to report the colors of their respective edges, consider a \emph{weaker} requirement that the color of each edge needs to be reported by  one of the players. The proof of \Cref{them:lower} can be modified so that we have $\Omega(n)$ bits of communication complexity for the modified version of  the $(2\Delta-1)$-edge coloring problem. The details will be discussed in \Cref{sec:lowerbound}. Moreover, one can show a reduction from the modified version of the $(2\Delta-1)$-edge coloring problem in the two-party communication model to the $(2\Delta-1)$-edge coloring problem in the $W$-streaming model. Hence, we have the following corollary:
%\gopinath{Is the above description fine here. Or shall we move them to \Cref{sec:lowerbound}?}

\begin{restatable}{corollary}{wstream}\label{coro:w}
    Any constant-pass {randomized} $W$-streaming algorithm that solves the $(2\Delta - 1)$-edge coloring problem requires $\Omega(n)$ bits of space in expectation.
\end{restatable}
There are a few works that deal with the upper bound for the edge coloring problem in the $W$-streaming model \cite{behnezhad2019streaming,charikar2021improved,ansari2022simple, DBLP:conf/icalp/SaneianB24}. As discussed earlier, the best-known upper bound in this context is an $O(\Delta)$-edge coloring randomized algorithm in the $W$-streaming model that uses space $\widetilde{O}(n\sqrt{\Delta})$~\cite{DBLP:conf/icalp/SaneianB24}. However, prior to our work, as stated in \Cref{coro:w}, no non-trivial space lower bound was known, even for the $(2\Delta-1)$-edge coloring problem in the $W$-streaming model.

\color{black}
%\gopinath{The above text is to take care of the second comment of the third reviewer.}

%\paragraph{Open question:} An interesting open question is to design protocols for $(\Delta + 1)$-vertex coloring and $(2\Delta - 1)$-edge coloring that utilize \(O(n)\) bits of communication cost while maintaining an asymptotically minimum number of rounds. Although our protocols are round-efficient, we do not know if they achieve the optimal number of rounds while maintaining \(O(n)\) bits of communication.\hung{Should I add another interesting open question is finding deterministic or high probability protocol?}\yijun{My original suggestion is to expand this paragraph into a section at the end of the paper, where we discuss about deterministic or high probability protocols, smaller number of colors, and maybe other open questions.}

\subsection{Additional related work}\label{sec:related}

%\paragraph{Graph problems in communication complexity:}
Graph problems have been a central topic in the exploration of communication complexity. Within this framework, numerous fundamental graph problems exhibit a range of upper and lower bounds, along with varying round complexities, all quantified in bits. For the {connectivity} problem, there exists an $O(n \log n)$-bit communication protocol where both players transmit the spanning forests of their respective subgraphs in a single round, resulting in an upper bound of $O(n \log n)$ bits. The lower bound for connectivity in this context is $\Omega(n)$ bits \cite{hajnal1988communication}, suggesting that the current upper bound is nearly optimal. \citet{assadi2021simple} demonstrated that the {minimum cut} problem can be addressed with an $O(n \log n)$-bit randomized communication protocol that operates in 2 rounds; however, the lower bound is notably stronger at $\Omega(n \log \log n)$ bits. This highlights a discrepancy between the upper and lower bounds regarding both communication complexity and round complexity. For the {maximum matching} problem in bipartite graphs, an established deterministic communication protocol requires $O(n \log^2 n)$ bits \cite{blikstad2022communication} and necessitates $n$ rounds. The lower bound for this problem stands at $\Omega(n)$ bits, as evidenced by the work of \cite{babai1986complexity,razborov1992distributional,ivanyos2012new}. Notably, all of the aforementioned problems have deterministic lower bounds of $\Omega(n \log n)$ bits \cite{hajnal1988communication,blikstad2022communication}. While the protocols discussed above exhibit near-linear randomized complexity and are efficient, bridging the gap between the nearly linear upper bounds and the linear randomized lower bounds, along with developing algorithms that are round-efficient, remains a significant open challenge in communication complexity. Recently, \citet{10.1145/3662158.3662796} revealed that $\Theta(n)$ bits of communication suffice for the $(\Delta+1)$-vertex coloring problem. However, the expected round complexity of their protocol is $O(n)$, which indicates that it is not round-efficient.

As noted previously, in one-pass streaming, the $(\Delta + 1)$-vertex coloring problem can be solved using $O(n \log^3 n)$ bits of space \cite{assadi2019sublinear}. Furthermore, a deterministic algorithm accomplishes the same coloring task with $O(n \log^2 n)$ bits of space and $O(\log \Delta \cdot \log \log \Delta)$ passes \cite{assadi2023coloring}, highlighting the balance between space efficiency and pass complexity. An algorithm for $(\Delta + 1)$-vertex coloring that performs $\widetilde{O}(n^{3/2})$ queries in the general graph property testing model has also been presented \cite{assadi2019sublinear}. The best-known algorithm for the $(\Delta + 1)$-vertex coloring (and its generalization, the $(\mathsf{degree}+1)$-list coloring) utilizes $O(\log^{5/3} \log n)$, $O(\log^3 \log n)$, and $O(1)$ rounds in the \LOCAL, \CONGEST, and \CCLIQUE{} models of distributed computing, respectively \cite{chang2020distributed,halldorsson2022near, GhaffariGrunau2024, DBLP:conf/stoc/HalldorssonKMT21,DBLP:conf/podc/HalldorssonNT22,chang2019complexity,czumaj2021deterministic,DBLP:conf/icalp/CoyCDM23}. In the context of \emph{massively parallel computation}, both the $(\Delta + 1)$-vertex coloring and $(\mathsf{degree}+1)$-list coloring can be solved in $O(\log \log \log n)$ rounds \cite{chang2019complexity,DBLP:conf/podc/CzumajDP21,DBLP:conf/ipps/CoyCDM24}. Additionally, there is a dynamic algorithm for the $(\Delta + 1)$-vertex coloring problem that achieves constant amortized update time \cite{DBLP:journals/talg/BhattacharyaGKL22,DBLP:journals/talg/HenzingerP22}.

As already mentioned, no efficient streaming algorithm exists for the $(2\Delta - 1)$-edge coloring problem, whereas a streaming algorithm for $O(\Delta)$-edge coloring utilizes $\widetilde{O}(n\sqrt{\Delta})$ space \cite{DBLP:conf/icalp/SaneianB24}. The $(2\Delta - 1)$-edge coloring problem has algorithms with round complexities of $O(\log^3 \log n)$ and $O(\log^4 \log n)$ in the \LOCAL{} and \CONGEST{} models of distributed computing, respectively \cite{DBLP:journals/siamcomp/Harris20,DBLP:journals/tcs/HalldorssonN23}. Additionally, a dynamic algorithm for the $(2\Delta - 1)$-edge coloring problem offers $O(\log \Delta)$ update time in the worst case \cite{bhattacharya2018dynamic}.

In the \CONGEST{} model of distributed computing, there is also a line of research studying algorithms that are both round-efficient and communication-efficient~\cite{MashreghiK17, HaeuplerHW18, GmyrP18, PanduranganRS17, DufoulonPPPPR24,DufoulonPRS24}.

%{While our primary focus is on obtaining protocols that are both round- and communication-efficient, a similar line pursued in the distributed setting (particularly in \CONGEST{}), where both tradeoff between round and message complexity are studied \cite{ MashreghiK17, HaeuplerHW18, GmyrP18, PanduranganRS17, DufoulonPPPPR24,DufoulonPRS24}. Our techniques, model, and problems differ, but this shared interest in understanding the interplay between multiple efficiency measures reflects a broader theme in distributed algorithm design.}

\subsection{Paper organization}
We provide a technical overview of our results in \Cref{sec:overview}. The preliminaries are presented in \Cref{sec:prelim}. Our protocols for the $(\Delta+1)$-vertex coloring problem and the $(2\Delta-1)$-edge coloring problem are presented in \Cref{sec:vertex,sec:edge-upper}, respectively. The lower bound on the communication complexity of the $(2\Delta-1)$-edge coloring problem is presented in \Cref{sec:lowerbound}. We conclude with open problems in \Cref{sect:conclusions}. Missing proofs from \Cref{sec:prelim} are provided in \Cref{app:col-samp,app:d1lc}.

%\gopinath{To do at the end.}
\section{Technical overview}\label{sec:overview}
%\yijun{Write something before the first paragraph. The sudden transition into palette sparsification is unnatural.}

We give an overview of the proofs of \Cref{thm:vertex,thm:edge,them:lower}.

\subsection{\texorpdfstring{$(\Delta+1)$}{(Delta+1)}-vertex coloring algorithm}
%\textcolor{blue}{In progress!}

First, we briefly discuss the approach of \citet{10.1145/3662158.3662796} for solving the $(\Delta+1)$-vertex coloring problem. Alice and Bob use public randomness to select a random ordering $\pi$ of the vertices. They then color the vertices sequentially according to this order. For a given vertex $v$, Alice and Bob must find an available color from $\{1, \ldots, \Delta+1\}$ that is not already assigned to any of $v$'s neighbors. This is nontrivial because the edges incident to $v$ are divided between Alice and Bob. However, if $v$ has $k_v$ available colors, an available color can be found with $O\left(\log^2 ((\Delta+1) / k_v)\right)$ expected bits of communication and $O\left(\log ((\Delta+1) / k_v)\right)$ rounds in expectation. A slight modification of this approach allows us to sample an available color uniformly at random with the same communication and round complexity. Therefore, note that, an available color for a vertex $v$ can be sampled with $O(\log^2 \Delta)$ bits of communication and $O(\log \Delta)$ rounds; in the worst case. Since the vertices are colored in random order, meaning $k_v$ is uniformly distributed over $\{\Delta+1, \ldots, \Delta+1-\deg(v)\}$, the expected communication cost is $O(1)$ bits and the expected number of rounds required is $O(1)$ to color a vertex $v$. This implies that the expected communication complexity of the protocol of Flin and Mittal is $O(n)$ bits and the expected round complexity is $O(n)$.

% The worst case communication comround complexity is $O(n \log \Delta)$ since we are coloring one vertex at a time and finding an available color for a vertex may take $O(\log \Delta)$ rounds in the worst case.

Our approach is based on the \emph{random color trial} technique, which has been applied in the context of $(\Delta+1)$-vertex coloring and related problems in \emph{distributed computing} \cite{JOHANSSON1999229,barenboim2016locality}. In each iteration of the random color trial, every uncolored vertex participates with probability \( {1}/{2} \) and attempts to color itself. Each participating vertex \( v \) selects a color uniformly at random from the set \( \{1, \dots, \Delta + 1\} \), ensuring that the color chosen is not assigned to any of its neighbors. If a vertex successfully colors itself—i.e., no neighbor has chosen the same color—it retains that color permanently. It is straightforward to observe that, given any partial coloring of a $(\Delta+1)$-vertex coloring instance, a vertex \( v \) will always have at least one more available color than the number of its uncolored neighbors. Putting things together, one can show that the probability that a vertex gets successfully colored in an iteration is a constant (see \Cref{lem:active}). Consequently, the probability that a vertex remains uncolored after the \( i \)th iteration decreases exponentially with \( i \). Thus, running the random color trial for \( O(\log n) \) iterations on a $(\Delta+1)$-vertex coloring instance guarantees a valid coloring with high probability. Specifically, the probability of obtaining a valid solution is at least \( 1 - 1/n^c \) for some large constant \( c \).

In our protocol (in the two-party communication model), we first run the random color trial for \( O(\log \log n) \) iterations. To implement each iteration, Alice and Bob must sample an available color uniformly at random for each uncolored vertex \( v \) in parallel. As discussed earlier, the expected communication cost and the expected number of rounds required for sampling an available color uniformly at random for a vertex \( v \) are \( O(\log^2((\Delta+1) / k_v)) \) and \( O(\log((\Delta+1) / k_v)) \), respectively, where \( k_v \) is the number of available colors for \( v \). On one hand, all vertices attempt to color themselves in parallel during each iteration. On the other hand, the values of \( k_v \) can differ across vertices, and in particular, \( k_v \) can be 1 for some vertices. This means that in the worst case, the number of rounds required to implement each iteration could be \( O(\log \Delta) \). Therefore, the total number of rounds needed to execute all \( O(\log \log n) \) iterations of the random color trial is \( O(\log \log n \cdot \log \Delta) \) in the worst case.

We show that the expected cost of attempting to color a vertex \( v \) over all \( O(\log \log n) \) iterations is \( O(1) \) bits. We divide the analysis into two cases:

\begin{description}
    \item[Case 1: \( \deg(v) \leq \Delta/2 \):]
    In this case, when \( v \) participates in an iteration, the number of available colors \( k_v \) for \( v \) is greater than \( \Delta/2 \). Thus, the communication cost to sample an available color for \( v \) is \( O\left(\log^2 \frac{\Delta + 1}{k_v}\right) = O(1) \) bits. Additionally, since the probability that a vertex remains uncolored in the \( i \)th iteration decreases exponentially with \( i \), we conclude that the expected communication cost due to attempting to color vertex \( v \) (with \( \deg(v) \leq \Delta/2 \)) is \( O(1) \) bits.

    \item[Case 2: \( \deg(v) \geq \Delta/2 \):]
    In this case, \( v \) has a large initial degree. Since an uncolored vertex participates in each iteration with probability \( {1}/{2} \), we expect that, for all \( i \in O(\log \log \Delta) \), the number of uncolored neighbors of \( v \) remains \emph{large} in the \( i \)th iteration. Specifically, the probability of finding an iteration \( i \) such that the number of uncolored neighbors is less than \( {\Delta}/{\alpha^i} \) is at most \( O\left({1}/{\log^2 \Delta}\right) \), where \( \alpha > 1 \) is a suitable constant. If the number of uncolored neighbors in the \( i \)th iteration is at least \( {\Delta}/{\alpha^i} \), then \( k_v \geq {\Delta}/{\alpha^i} \), and the expected communication cost to sample an available color is \( O\left(\log^2 \frac{\Delta + 1}{k_v}\right) = O(i^2) \) bits.

    Given that the probability of \( v \) remaining uncolored in the \( i \)th iteration decreases exponentially with \( i \), and that the worst-case communication cost to sample an available color is \( O(\log^2 \Delta) \), the total expected communication cost over \( O(\log \log \Delta) \) iterations is \( O(1) \) bits. Furthermore, the expected communication cost after \( O(\log \log \Delta) \) iterations remains \( O(1) \), due to the fact that the probability of \( v \) being uncolored after these iterations is \( O\left({1}/{\log^2 \Delta}\right) \), combined with the worst-case communication cost of \( O(\log^2 \Delta) \) bits.
\end{description}

After the random color trials, we obtain a valid partial coloring of the graph \( G \). Let \( Z \) denote the set of uncolored vertices. Using the fact that the probability a vertex remains uncolored after the \( i \)th iteration decreases exponentially with \( i \), we conclude that the expected size of \( Z \) is \( O\left( {n}/{\log^4 n} \right) \). Although the leftover coloring problem constitutes a valid input coloring instance, it may not necessarily correspond to a \( (\Delta + 1) \)-coloring instance. Instead, it represents a \( (\mathsf{degree} + 1) \)-list coloring (\DLC) instance on the vertex set \( Z \). In this setting, each vertex \( v \in Z \) has a list of available colors based on the colors already assigned to its neighbors during the random color trials. In the \DLC problem, for a graph \( G = (V, E) \), each vertex has a palette of acceptable colors that exceeds its degree by at least one. The objective is to find a proper coloring using these palettes.

To color the leftover instance, we utilize the \emph{palette sparsification theorem} for \DLC, as presented by \citet{halldorsson2022near}. Alice and Bob, for each vertex \( v \in Z \), reduce the set of acceptable colors for \( v \) by independently sampling \( O(\log^2 |Z|) \) colors from its palette. This step requires \( O\left( |Z| \cdot \log^2 |Z| \cdot \log^2 \Delta \right) = O\left( |Z| \cdot \log^4 n \right) \) bits of communication and can be completed in \( O(\log \Delta) \) rounds. Next, Alice and Bob remove any edge \( \{u, v\} \) where the lists of \( u \) and \( v \) are disjoint. According to the palette sparsification theorem, the sparsified instance will have \( O\left( |Z| \cdot \log^2 |Z| \right) \) edges and is colorable with a probability of at least \( 1 - {1}/{|Z|^c} \), for some large constant \( c \). Since the expected size of \( |Z| \) is \( O\left( {n}/{\log^4 n} \right) \), this implies that the protocol for coloring the leftover \DLC instance on the vertex set \( Z \) has an expected communication cost of \( O(n) \) bits and requires \( O(\log \Delta) \) rounds in the worst case.

% By leveraging this theorem, we can show that the remaining instance can be solved with an expected communication cost of $O(n)$ bits and within $O(\log \Delta)$ rounds. This approach completes the coloring of the uncolored vertices in $Z$, ensuring that the overall coloring process remains efficient in both communication and round complexity.

%\textcolor{blue}{Here, may be we talk briefly about the algorithm of \cite{10.1145/3662158.3662796}. Then also talk very briefly about one round $O(n\log^3 n)$ protocol using PST. Then discuss our idea using random color trial and PST for list coloring problem.}\gopinath{?}

\subsection{Upper bound for \texorpdfstring{$(2\Delta-1)$}{(2Delta-1)}-edge coloring}
%We briefly discuss two well-known existing results in edge coloring. 
According to Vizing's theorem \cite{vizing1964}, the number of colors needed to edge-color a simple graph is either its maximum degree \(\Delta\) or \(\Delta + 1\). An extension of Vizing's theorem by \citet{fournier1973colorations} states that if the vertices of maximum degree \(\Delta\) form an independent set, then the number of colors needed for edge coloring is \(\Delta\). These existential results are useful in designing protocols for edge coloring. As we have already discussed, in the protocol for the \((2\Delta)\)-edge coloring problem in \Cref{thm:2delta}, we partition the set of $2\Delta$ colors into two halves of equal size: Alice’s palette and Bob’s palette, and then both parties construct the coloring of \citet{fournier1973colorations} locally.

The $(2\Delta-1)$-edge coloring problem is much more challenging. We again partition the set of \(2\Delta - 1\) colors into three subsets: two sets of size \(\Delta - 1\) each (Alice's palette and Bob's palette) and one \textsf{special} color. To apply the aforementioned coloring of \citet{fournier1973colorations}, we must ensure that the maximum degree of the graphs for Alice and Bob is \(\Delta - 1\) and that the vertices of degree \(\Delta - 1\) form an independent set. Following the same approach as for the \((2\Delta)\)-edge coloring (\Cref{thm:2delta}), we \emph{sequentially} identify and defer the coloring of edges connecting high-degree vertices (those with degree at least \(\Delta - 1\)), coloring the deferred edges with colors from the other party's palette. We consider a subgraph consisting of all the deferred edges, which, by simple reasoning, has a maximum degree of \(2\). However, we are not finished yet; in one party's remaining subgraph, there may still be vertices of degree \(\Delta\)  that are originally ``independent''—not connected to any other high-degree vertex. Therefore, the maximum degree of the remaining subgraph may still be \(\Delta\), and we cannot \((\Delta - 1)\)-edge color it.

To address this issue, we find a \emph{matching} that covers all the ``independent'' vertices of degree \(\Delta\), coloring the edges in this matching with the \textsf{special} color. However, there may be cases where two edges in the respective matchings of both parties are incident to the same vertex, preventing us from coloring all matching edges with the \textsf{special} color. In such situations, we use colors from the other party's palette. 

Consider a vertex that is incident to these two edges in both parties's respective matchings; it must have a degree of at most \(\Delta / 2\) in either party. Therefore, the number of colors occupied by edges incident to that vertex in the other party's graph is at most \(\Delta / 2\), leaving nearly half of the palette available. There exists a deterministic protocol that allows one party to know an available color from another party's palette using $O(n)$ bits and $O(1)$ rounds. This allows us to color all edges in the matching using $O(n)$ bits and $O(1)$ rounds of communication.%\gopinath{This paragraph needs to be modified.}\hung{Done}

The remaining subgraph then satisfies the conditions of the existing theorem by \citet{fournier1973colorations}: the vertices of maximum degree \(\Delta - 1\) form an independent set, allowing each party to use their respective palettes to color their subgraphs. Regarding the deferred edges, we observe that each vertex in this subgraph can be incident to at most one edge in the other party's graph and at most one edge in the matching. Additionally, the deferred subgraph has a maximum degree of \(2\). Thus, using a constant number of colors from the other party's palette is sufficient to color this deferred subgraph, and communicating this constant number of colors for each vertex can be accomplished with \(O(n)\) bits in $O(1)$ rounds of communication.

%\gopinath{Though achieving $O(n)$ bits of communication is the main goal, a brief discuss about the round complexity would be great.}\hung{Yes}

\subsection{Lower bound for \texorpdfstring{$(2\Delta-1)$}{(2Delta-1)}-edge coloring}

We discuss the lower bound proof for the $(\Delta+1)$-vertex coloring problem from \citet{10.1145/3662158.3662796}, which uses a reduction from the \emph{learning} problem, known to have a lower bound of $\Omega(n)$ bits on communication complexity. In this context, Alice holds a string $\mathbf{x} = x_1 x_2 \ldots x_n \in \{0,1\}^n$, and Bob's goal is to learn this string. For each bit $x_i$, a gadget is constructed using four vertices $a_i, b_i, c_i$, and $d_i$. The edges $\{a_i, b_i\}$ and $\{c_i, d_i\}$ are always present, while additional edges depend on the value of $x_i$: if $x_i = 0$, the edges $\{a_i, c_i\}$ and $\{b_i, d_i\}$ are included; if $x_i = 1$, the edges $\{a_i, d_i\}$ and $\{b_i, c_i\}$ are included. Alice possesses all $4n$ edges, while Bob has none, and the maximum degree of the graph is $2$. After executing the $(\Delta+1)$-vertex coloring protocol, both parties obtain a  valid 3-vertex coloring. It is not difficult to show that a $3$-vertex coloring enables Bob the recovery of all edges and revealing each $x_i$. Consequently, Bob learns $\mathbf{x}$, establishing a lower bound of $\Omega(n)$ bits on the communication complexity of the $(\Delta+1)$-vertex coloring problem.

%This implies that, if the graph obtained from $\mathbf{x}$ is the 

%We first discuss the lower bound proof for the $(\Delta+1)$-vertex coloring problem by \citet{10.1145/3662158.3662796}, and explain why we cannot naively extend the idea to prove the lower bound of $(2\Delta - 1)$-edge coloring. The proof by \citet{10.1145/3662158.3662796} uses a reduction from vertex coloring to sending $n$-bits string.\gopinath{The problem from which they reduce is that the objective of Bob is to learn the string.} Given $n$-bits string $x=(x_1,x_2,\dots,x_n)$, the referee then constructs a \textsf{special} degree-2 graph with $4n$ vertices, and partitioned the edges to Alice and Bob, such that if Bob knows a 3-vertex coloring of the graph, then Bob can recover the actual value of $x$. Specifically, for each bit $x_i$, the referee constructs a 4-vertex graph with degree 2, sends a set of edges to Alice based on $x_i$ (If $x_i$=0, the referee sends $E_1$. Else, $E_2$ will be sent), whereas Bob will always get the same set of edges regardless of $x_i$. This is constructed such that for edges $E_1$ and $E_2$ given to Alice, the 3-coloring of the 4-vertices graph will always be disjoint. Therefore, if Bob knows the 3-coloring, he can recover the bit $x_i$. 
Note that the graph constructed in the above reduction consists of the union of $ n $ copies of $ C_4 $, which has a maximum degree of 2. It is important to observe that $ C_4 $ is edge-vertex dual to itself, and for $ \Delta = 2 $, we have $ 2\Delta - 1 = \Delta + 1 $. While it might initially seem that a similar graph could be employed to establish a lower bound for the $ (2\Delta - 1) $-edge coloring, this approach cannot be directly applied. There is a fundamental distinction between the $ (\Delta + 1) $-vertex coloring problem and the $ (2\Delta - 1) $-edge coloring problem within the context of two-party communication complexity. In the $ (\Delta + 1) $-vertex coloring problem, both Alice and Bob must be aware of the colors of all the vertices in the graph. In contrast, for the $ (2\Delta - 1) $-edge coloring problem, each party only needs to know the colors of their respective edges. In the construction for the lower bound of the $ (\Delta + 1) $-vertex coloring problem, Alice possesses all the edges while Bob has none. Consequently, the solution for the $ (2\Delta - 1) $-edge coloring problem is straightforward for Alice to obtain on her own, as she does not require any information about edge colors from Bob, and also Bob does not need to know the colors of the edges of Alice. Additionally, the $ (2\Delta - 1) $-edge coloring problem is inherently \emph{easier} to solve in the following sense: a party can construct a partial solution to the edge coloring problem based solely on their own input set of edges, thereby allowing the coloring to be extended to the entire graph. However, this is not necessarily the case in the $ (\Delta + 1) $-vertex coloring problem.

Our approach to prove the lower bound constructs a graph of constant size with $\Delta=2$ where any algorithm that finds a $(2\Delta - 1)$ edge coloring (ie, $3$ edge coloring) with zero communication protocol has a constant success probability smaller than $1$. Considering $n$ independent instances of the above constant-size graph, we have an $O(n)$-vertex graph with $\Delta=2$, where the probability of any zero-communication protocol producing a valid $(2\Delta - 1)$-edge coloring is at most $2^{-\Omega(n)}$. This follows from the \emph{parallel repetition theorem} of Raz~\cite{raz1998parallel, Holenstein_2009}. Conversely, if we assume the existence of a constant-error randomized algorithm $\mathcal{P}$ that solves the $(2\Delta - 1)$-edge coloring problem with $o(n)$ bits in the worst case, a zero-communication protocol could be derived by having both parties independently guess the communication pattern of $\mathcal{P}$. This would yield a success probability of $2^{-o(n)}$, contradicting the previous upper bound of $2^{-\Omega(n)}$ on the success probability and thus establishing the lower bound. %{The approach of first considering a constant-size hard instance and then combining $\Omega(n)$ many such independent instances to get the hard instance to get the desired lower bound for the problem of consideration have been recently considered by \citet{konrad2021robust} as well as \citet{flin2024decentralized}.}

{The technique of constructing a constant-size hard instance and then combining $\Omega(n)$ independent copies to construct a larger hard instance has recently been used by \citet{konrad2021robust} and \citet{flin2024decentralized} to derive the desired lower bound for the problem of interest. Moreover, the construction of the lower bound of \citet{10.1145/3662158.3662796} for the $(\Delta+1)$-vertex coloring problem also follows the approach of combining $n$ independent hard instances of constant size. Thus, the main challenge in this approach lies in constructing a hard instance of constant size.}

\paragraph{The constant-size hard instance:} Consider the set of vertices $ V = \{v_A, v_B, v_1, \ldots, v_7\} $. Alice selects two edges from the possible set of seven edges $ \{\{v_A, v_i\} : 1 \leq i \leq 7\} $, while Bob similarly selects two edges from $ \{\{v_i, v_B\} : 1 \leq i \leq 7\} $. Notably, the maximum degree of the graph is 2. The objective for both Alice and Bob is to find a 3-coloring of the edges using colors $ \{c_1, c_2, c_3\} $. In the context of any deterministic algorithm performing this 3-edge coloring, we label each vertex $ v_i \in \{v_1, \ldots, v_7\} $ with a subset from $ \{c_1, c_2, c_3\} $ based on the algorithm's behavior. It is important to note that the execution of the algorithm from Alice's perspective is independent of that from Bob's, as there is no communication between them. Let $ \labela(v_i) $ and $ \labelb(v_i) $ represent the labels of $ v_i $ with respect to Alice and Bob, respectively. For each vertex $ v_i $, we add color $ c_j \in \{c_1, c_2, c_3\} $ to $ \labela(v_i) $ if there exists an input set of two edges that Alice could have such that the algorithm subsequently colors the edge $ \{v_A, v_i\} $ with color $ c_j $. Similarly, we include color $ c_j $ in $ \labelb(v_i) $ if there is a potential input set of two edges for Bob that results in the algorithm coloring the edge $ \{v_i, v_B\} $ with color $ c_j $. If the labels of two vertices $ x $ and $ y $ with respect to one party (say Alice) are the same singleton set, then for the input set of edges $ \{v_A, x\} $ and $ \{v_A, y\} $, the algorithm does not produce a valid coloring. Since there are three colors, there can be at most three vertices with singleton sets as their labels with respect to one party. Therefore, there must exist a vertex $ v_i $ such that $ \labela(v_i) $ and $ \labelb(v_i) $ contain at least two colors and share a common color $ c_j $. This implies there exists an input for Alice where $ \{v_A, v_i\} $ is one edge, and Alice colors $ \{v_A, v_i\} $ with color $ c_j $. Furthermore, there exists an input for Bob where $ \{v_i, v_B\} $ is one edge, and Bob colors $ \{v_i, v_B\} $ with the same color $ c_j $. Since Alice's and Bob's executions are independent, there is a joint input for both such that both $ \{v_A, v_i\} $ and $ \{v_i, v_B\} $ are edges in the input graph, and they are both colored with $ c_j $ by the algorithm. This results in a 3-edge coloring that is not proper. Although the preceding discussion addresses the deterministic lower bound, the argument extends to the randomized case as follows: both Alice and Bob are provided with two edges from the sets $ \{\{v_A, v_i\} : 1 \leq i \leq 7\} $ and $ \{\{v_i, v_B\} : 1 \leq i \leq 7\} $, respectively, chosen uniformly at random. Moreover, we add a color $ c_j $ to $ \labela(v_i) $ (or $ \labelb(v_i) $) if there exists an input set of edges for Alice (or Bob) such that the algorithm colors $ v_i $ with color $ c_j $ with a probability at least $ 1/5 $. The details of this argument are deferred to \Cref{sec:lowerbound}.

\section{Preliminaries}\label{sec:prelim}
% \paragraph{Notation.} Throughout this paper, we will work with input graphs \( G = (V, E) \) on the vertex set \( V = [n] := \{1, 2, \ldots, n\} \), with maximum degree \( \Delta \). Let \( N(v) \) denote the neighborhood of \( v \) in \( G \), and \( d_v = |N(v)| \) its degree. For any integer \( q > 1 \), a \( q \)-coloring of \( G \) is a vector \( C \in [q]^n \). We say the coloring is proper if for all edges \( \{u, v\} \in E \) we have \( C(u) \neq C(v) \).

\subsection{Notations and model} For \( t \in \mathbb{N} \), we define \( [t] := \{1, \ldots, t\} \). An undirected simple graph is denoted by \( G = (V, E) \), where \( V \) is the vertex set and \( E \) is the edge set. If there is an edge between vertices \( u \) and \( v \), it is denoted by \( \{u, v\} \). For a vertex \( v \in V \), \( N(v) \) and \( \deg(v) \) denote the set of neighbors and the degree of \( v \) in \( G \), respectively. The maximum degree of any vertex in \( G \) is denoted by \( \Delta \). In the two-party communication model, Alice and Bob receive the parameters \( n \) and \( \Delta \). The vertices of the graph $G$ are known to both parties. The edges of graph \( G \) are partitioned adversarially between the two parties, denoting the edges assigned to Alice as \( E_A \) and those assigned to Bob as \( E_B \). The neighborhoods of a vertex \( v \) in Alice's and Bob's respective graphs are defined as \( N_A(v) = \{u \in V : \{u, v\} \in E_A\} \) and \( N_B(v) = \{u \in V : \{u, v\} \in E_B\} \), respectively, such that \( N(v) = N_A(v) \sqcup N_B(v) \). In the \( (\Delta + 1) \)-vertex coloring problem, both Alice and Bob are required to output the coloring of all the vertices. However, in the \( (2\Delta - 1) \)-edge coloring problem, the objective for Alice and Bob is to output the coloring of only their respective edges. Although both parties have access to public/shared randomness, this requirement can be relaxed to the use of private randomness with only an additional \( O(\log n + \log(1/\delta)) \) bits in communication cost, where \( \delta \) represents the upper bound on failure probability \cite{newman1991private}.%\gopinath{May be the last part will goto a different preliminary.}

% \subsection{Concentration Inequalities}
% %\cite{dubhashi22concentration}

% \begin{proposition}[Chernoff Bound] Let \( X_1, X_2, \ldots, X_n \) be independent random variables in \([0, 1]\) and \( X = \sum_{i=1}^n X_i \). Then, for any \(\epsilon > 0\),
% \[
% \Pr\left[X < (1 + \epsilon) \mathbb{E}[X]\right], \ 
% \Pr\left[X > (1 + \epsilon) \mathbb{E}[X]\right] \leq \exp\left( -\frac{\epsilon^2 \mathbb{E}[X]^2}{2 + \epsilon} \right).
% \]\yijun{Check if this is the version of Chernoff bound we want. I think even our first use of Chernoff bound is not compatible with the version stated above.}\gopinath{As this is the only concentration we are using, may be we state while using for the first time}
% \end{proposition}

% \yijun{Not enough paraphrasing for the writing above in this section.}
% \begin{proposition} Basic inequality:
% $1 - x \geq e^{-\frac{x}{1-x}} \quad \text{if} \quad x < 1$
% \end{proposition}

\subsection{Sampling an available color uniformly}\label{sec:col-samp}

While solving the $(\Delta+1)$-vertex coloring problem, let us consider that we have a partial proper vertex coloring of the graph $G$. Note that both Alice and Bob know the colors assigned to the subset of already colored vertices. For an (uncolored) vertex $v$, let $A$ and $B$ denote the set of colors assigned to some vertex in $N_A(v)$ or $N_B(v)$. We refer to each color in the set $[\Delta + 1] \setminus (A \cup B)$ as an \emph{available} color for $v$. In our vertex-coloring algorithm, an important step for Alice and Bob is to \emph{sample} an {available} color uniformly at random (such that the sampled color is known to both parties without further communication). %In this section, we establish that this task can be achieved efficiently (see \Cref{lm:ksi}).

In \citet{10.1145/3662158.3662796}, the authors address the problem of \emph{finding} an arbitrary available color for a vertex. We adapt their method to instead sample an available color \emph{uniformly at random}. The key observation is that their procedure does not inherently favor any particular available color. To achieve uniform sampling, Alice and Bob first apply a random permutation to the set of colors $[\Delta + 1]$ using public randomness. They then run the algorithm from Flin and Mittal to find an available color. Since the permutation ensures that all available colors are equally likely, the resulting color is chosen uniformly at random from the set of available colors. This leads to the following lemma.

\begin{lemma}[\textbf{Sampling an available color uniformly}]
    \label{lm:ksi}
Consider the setting where we have a partial proper vertex coloring of $G$ such that both Alice and Bob know the colors assigned to already colored vertices. For a vertex $v$, let $k$ with $1 \leq k \leq \Delta+1$ denote the number of available colors. There exists a randomized protocol \colorsamp for both Alice and Bob to {sample an available color uniformly at random} such that the following conditions are satisfied.
\begin{enumerate}
    \item[(i)] The sampled color is known to both the parties.
    \item[(ii)] In expectation, it requires \( O\left(\log^2 \frac{\Delta +1 }{k}\right) \) bits of communication, In the worst case, it requires $O(\log^2 (\Delta+1))=O(\log^2 \Delta)$ bits of communication.
    \item[(iii)] In expectation, it runs in \( O\left(\log  \frac{\Delta + 1}{k}\right) \) rounds. In the worst case, it requires $O(\log (\Delta+1))=O(\log \Delta)$ number of rounds.
\end{enumerate}
\end{lemma}
For completeness, we discuss the procedure of Flin and Mittal to find an arbitrary color and establish the above lemma in  \Cref{app:col-samp}.

\subsection{\texorpdfstring{$(\mathsf{degree}+1)$}{(degree+1)}-list coloring in the two party communication model}\label{sec:dlc}

%In our $(\Delta+1)$-vertex coloring algorithm, after applying the randomized color trial, we have a partial vertex coloring. Most importantly, the leftover instance is also colorable and contains $O(n/\log^ 4n)$ vertices in expectation. However, it is not necessarily a $(\Delta+1)$-vertex coloring instance. Rather, it is a \textsf{(degree}+1)-list coloring (\DLC) instance (to be described below).% Here, we discuss a randomized protocol for $\DLC$ in the two-party communication model using the \emph{Palette Sparsification Theorem} for \DLC \cite{halldorsson2022near}.

The $\mathsf{(degree}+1)$-list coloring problem (\DLC)   is a natural generalization of the $(\Delta+1)$-coloring problem. Here, for a given graph $G = (V,E)$, each vertex $v$ of degree $\deg(v)$ in $G$ has  an input palette $\Psi(v)$ of acceptable colors such that $|\Psi(v)|\geq \deg(v)+1$, and the objective is to find a   vertex coloring such that vertex $v$ is assigned a color from $\Psi(v)$. %This says that, staring from an instance of $\DLC$ for graph $G$, we can generate a sparsified instance for general list coloring on  graph $H$ (on the same same set of vertices as $G$ but containing possibly \emph{few} edges) such that the instance is colorable. 

\paragraph{The set up for \DLC in the two-party communication model:} The edges of the input graph $G$ are partitioned among Alice and Bob.  For each vertex $v$, Alice has a list $\Psi_A(v)\subset [\Delta+1]$ of colors  and Bob has a list $\Psi_B(v) \subseteq [\Delta+1]$ of colors such that the palette of available colors for $v$ is $\Psi(v)= \Psi_A(v)\cap \Psi_B(v)$. Moreover, $\left|\Psi(v)\right|\geq \deg(v)+1$ and the objective is to find a vertex coloring such that each vertex $v$ is assigned a color from $\Psi(v)$.

 Now, we state the palette sparsification theorem for the \DLC problem~\cite{halldorsson2022near}, which extends the palette sparsification theorem for the $(\Delta+1)$-vertex coloring by \citet{assadi2019sublinear}.

\begin{proposition}[\textbf{Palette Sparsification 
 Theorem} \cite{halldorsson2022near}]\label{pro:pal}
Consider the \DLC problem on a graph $G$ on $n$ vertices as defined above. Suppose that for every vertex \( v \in V \), we independently sample a set \( L(v) \) of size \( l = \Theta(\log^2 n) \) uniformly at random from the colors in \( \Psi(v) \). Now consider a subgraph $H$ of $G$ on the vertex set and remove all edges $\{u,v\}$ such that $L(u)$ and $L(v)$ are disjoint. Then, the graph $H$ is vertex colorable such that each vertex $v$ is assigned a color from $L(v)$ and $H$ contains $O(n \log^2 n)$ edges, with probability $1-1/n^c$, where $c$ is a large constant.

\end{proposition}
% Then, with high probability, there exists a   coloring of \( G \) using the colors from the lists \( L(v) \) for \( v \in V \).

%     After sampling $\Theta(\log^2 n)$ colors at each vertex, we remove edges between vertices whose sampled colors sets are disjoint. This can be seen as a ``sparsification'' since the number of edges whose sampled color palette intersects is only $O(n \cdot \log^2 n)$ with  probability $1-1/n^c$, where $c$ is a sufficiently large constant.
% \end{theorem}

Using \Cref{pro:pal} and \Cref{lm:ksi}, we can design a randomized protocol for \DLC as stated in the following lemma.

\begin{lemma}[\textbf{Protocol for \DLC}] \label{lem:d1lc}
    There exists a randomized protocol that solves \DLC on a graph $G$, such that:
\begin{enumerate}
    \item[(i)] The protocol incurs $O\left(n \log^2 n \log^2 \Delta+ n\log^3 n\right)$ bits of communication in expectation.
    \item[(ii)] The protocol runs for $O(\log \Delta)$ rounds in the worst case.
\end{enumerate}

\end{lemma}
We prove the above lemma in \Cref{app:d1lc}.

\subsection{Classical existential results for edge coloring}
\label{sec:edge-existent}

In graph theory, Vizing's theorem asserts that any simple undirected graph can be edge colored with \emph{at most one color more} than its maximum degree $\Delta$. At least $\Delta$ colors are always required, allowing undirected graphs to be divided into two categories: ``class one'' graphs, which can be colored with $\Delta$ colors, and ``class two'' graphs, which need $\Delta + 1$ colors. While several authors have proposed additional criteria to determine if some graphs belong to class one or class two, a complete classification has not been achieved. \citet{fournier1973colorations} demonstrated that if the vertices with maximum degree $\Delta$ in a graph $G$ form an independent set, then $G$ must be classified as class one.

\begin{proposition}[\citet{vizing1964}]
Every simple undirected graph may be edge colored using a number of colors that is at most one larger than the maximum degree \( \Delta \) of the graph.
\end{proposition}

\begin{proposition}[\citet{fournier1973colorations}]
\label{thm:independent}
For any graph $G$, if the vertices of the maximum degree $\Delta$ form an independent set, there exists a proper edge-coloring with $\Delta$ colors.%\yijun{Find a paper to cite (the wiki page for Vizing theorem mentioned this version) }
\end{proposition}

\section{Protocol for \texorpdfstring{$(\Delta+1)$}{(Delta+1)}-vertex coloring} \label{sec:vertex}
In this section, we present a protocol for the $(\Delta+1)$-vertex coloring problem that uses $O(n)$ bits of communication and requires $O(\log \log n \cdot \log \Delta)$ rounds, thus proving \Cref{thm:vertex}. In \Cref{sec:col-trial}, we discuss the algorithm \randcolortrial, which colors a significant number of vertices, leaving only $O\left({n}/{\log^4 n}\right)$ uncolored in expectation. Furthermore, \randcolortrial communicates $O(n)$ bits in expectation and operates within $O(\log \log n \cdot \log \Delta)$ rounds in the worst case. In Sections \ref{sec:active} and \ref{sec:comm}, we analyze \randcolortrial. Lastly, in \Cref{sec:final-vertex-col}, we describe our full algorithm for the $(\Delta+1)$-vertex coloring, which combines \randcolortrial (from \Cref{sec:col-trial}) and the protocol for \DLC (as discussed in \Cref{lem:d1lc}).

\subsection{Random color trials in the two party communication model}\label{sec:col-trial}

%\gopinath{May be we make random color trial a subsection. Palette sparsification theorem another subsection. In another subsection, we do the final proof of the main theorem combining everything. The title of the section can be Protocol for $(\Delta+1)$-vertex coloring.}

% We assume that each vertex has a degree of \( \Delta \). By adding dummy edges to the graph, the algorithm will still provide a valid coloring of the original graph.

% \hung{Now, I don't know if we can assume this, I will think about it.} 
% \yijun{I think we cannot assume this. And why do we need this anyway?}

In this section, we discuss \randcolortrial (\Cref{alg:RandomColorTrial}) that is completed in $O(\log \log n)$ iterations such that each iteration spends $O(\log \Delta)$ rounds in the worst case. In the first iteration all vertices are \textsf{active}.
In each iteration, each (\textsf{active}) vertex has a probability of \( 1/2 \) to be \textsf{idle} and a probability of \( 1/2 \) to be \textsf{awake}. Upon \textsf{awake}, the vertex selects a color uniformly at random from its set of \emph{available} colors using \colorsamp as discussed in \Cref{lm:ksi}. If none of its neighbors choose the same color, the vertex permanently adopts this color and is marked as \textsf{done}. If a vertex is not marked as \textsf{done} at the end of the round, we classify it as \textsf{active} for the subsequent iteration.

%In each iteration, each vertex has a probability of \( 1/2 \) to \textsf{idle} and a probability of \( 1/2 \) to \textsf{awake}. Upon \textsf{awake}, the vertex selects a color uniformly at random from its set of \emph{available} colors using the \Cref{alg:k-Slack-Int} discussed above. If none of its neighbors choose the same color, the vertex permanently adopts this color and is marked as \textsf{done}. If a vertex is not marked as \textsf{done} at the end of the round, we classify it as \textsf{active}.

\begin{algorithm}[H]
\caption{\randcolortrial}
%\caption{An Expected $O(n)$ Communication Protocol for Random Color Trial}
\label{alg:RandomColorTrial}
\KwIn{A graph \( G \) with edges partitioned adversarially between two parties.}
\KwOut{A partial vertex coloring of the vertices in $G$ with expected number of uncolored vertices is at most \( \comments{{n}/{\log^4 n}} \).}
Mark each vertex as \textsf{active}.

\For{iteration \( i = 1, 2, \ldots, \comments{\lceil 1+4\log_{24/23} \log n \rceil} \)}{
    \For{each \textsf{active} vertex \( v \in G \) with \( k \) available colors}{
        \tcp{\textsf{in parallel}}
        Draw \( r \) using public randomness such that $\Pr[r=1]=\Pr[r=0]=1/2$;
        
        \tcp{Both parties have the same \( r \)}  
        \If{$r = 1$}{ \tcp{\textsf{awake}} 
            Select a color \( c \) uniformly at random from the set of available colors for vertex \( v \) using algorithm \colorsamp (as stated in \Cref{lm:ksi}) \tcp{Both parties have the same \( c \)}
            \If{none of \( v \)'s neighbors have chosen \( c \)}{
                Send a one-bit confirmation to the other party.\;
            }
            \If{none of \( v \)'s neighbors in the other party have chosen \( c \)}{
                Mark \( v \) as \textsf{done} and adopt color \( c \).\;
            }
        }
    }
}
\end{algorithm}
%In the following lemma, we formally have the guarantee from \randcolortrial.
\begin{lemma}[\textbf{Guarantee from \randcolortrial}]\label{lm:main}
    At the conclusion of the \randcolortrial procedure, a partial vertex coloring of the graph $G$ is produced. Furthermore, the following holds:
    \begin{enumerate}
        \item[(i)] The expected number of \textsf{active} (uncolored) vertices at the end of the procedure is \( O\left({n}/{\log^4 n}\right) \).
        
        \item[(ii)] The expected communication complexity is \( O(n) \).
        
        \item[(iii)] The worst-case round complexity is \( O(\log \log n \cdot \log \Delta) \).
    \end{enumerate}
\end{lemma}
Note that whenever we assign a color $c$ to vertex $v$ in $\mathsf{Random\mbox{-}Color\mbox{-}Trial}$, $c$ is used neither by any vertex in $N_A(v)$ nor in $N_B(v)$. Therefore, $\mathsf{Random\mbox{-}Color\mbox{-}Trial}$ produces a valid coloring at the end of the procedure.
We prove the claims (i)--(iii) in the subsequent subsections.

% We aim to show that when the number of \textsf{active} vertices is \( O\left(\frac{n}{\log^4 n}\right) \), the following results hold:

% \begin{enumerate}[label=\textbf{\arabic*.}, left=30pt, itemsep=1em]
%     \item The number of iterations is \( O(\log \log n) \) with high probability.
%     \item The expected bit complexity of the procedure is \( O(n) \).
%     \item The complexity is \( O(\log \log n \cdot \log \Delta) \) with high probability.
% \end{enumerate}

% We aim to demonstrate that after \( O(\log \log n) \) iterations, the following results hold:

% \begin{enumerate}[label=\textbf{\arabic*.}, left=30pt, itemsep=1em]
%     \item The expected number of \textsf{active} vertices is \( O\left(\frac{n}{\log^4 n}\right) \).
%     \item The expected bit complexity of the procedure is \( O(n) \).
%     \item The worst-case round complexity is \( O(\log \log n \cdot \log \Delta) \).
% \end{enumerate}

% Thus, from the above idea, sampling a set of \emph{available} palette is important. Since the edges are partitioned to two parties, the neighbors set of vertex $v$ in two parties, $N_A(v)$ and $N_B(v)$ , are disjoint. Since this is $\left(\Delta + 1\right)$-coloring, we have an important observation which is for a vertex $v$, there are $x$ number of its neighbor \textsf{idle} in round $i$, then the size of the \emph{available} palette of vertex $v$ in that round is at least $x$.
\subsection{Number of uncolored vertices after random color trial}\label{sec:active}

%We show that the probability of a vertex $v$ being \textsf{active} in iteration $i$ decreases exponentially. Then using the fact that \randcolortrial executes for $O(\log \log n)$ iterations and applying linearity of expectation, we get \Cref

To compute the expected number of uncolored vertices, we first focus on computing the probability that  a vertex \( v \) is active in round $i$. Let \( A_{i} \) denote the indicator random variable such that \( A_{i} = 1 \) if the vertex \( v \) is \textsf{active} in the \( i \)th iteration, where \( i \in \mathbb{N} \).
First, we show that the probability that an \textsf{active} vertex in the $i$th iteration is not \textsf{active} in the next iteration is at least a constant.

\begin{lemma}\label{lem:active}
    \(\Pr[A_{i+1} = 0 \mid A_{i} = 1] \geq \frac{1}{24}\), where $i \in \mathbb{N}$.
\end{lemma}

\begin{proof}
    Let \( v \) be an \textsf{active} vertex with degree \( d_v \) in the \emph{induced subgraph of \textsf{active} vertices} in the \( i \)th iteration. Define the random variable \( X \) as the number of awake neighbors of vertex \( v \) in one round. We have

    \[
    \E[X] = \frac{d_v}{2}.
    \]
    
    By applying Markov's inequality, we obtain
    
    \[
    \Pr\left[X \geq \frac{3}{4} d_v\right] \leq \frac{\E[X]}{\frac{3}{4} d_v} = \frac{2}{3}.
    \]
    
    This implies that
    
    \[
    \Pr\left[X < \frac{3}{4} d_v\right] \geq \frac{1}{3}.
    \]
    
    Given that \( X < \frac{3}{4} d_v \), the number of \textsf{idle} neighbors of vertex \( v \) is at least \( \frac{1}{4} d_v \). With at least \( \frac{1}{4} d_v \) neighbors \textsf{idle}, there are at least \( \frac{1}{4} d_v + 1 \) colors still available for vertex \( v \) to choose from, allowing it to successfully mark itself as \textsf{done} in that iteration. The probability of success for vertex \( v \) can be expressed as follows:
    
    \begin{align*}
        \Pr\left[\text{$v$ is \textsf{done}}\right] & = \Pr\left[\text{$v$ is \textsf{awake}}\right] \cdot \Pr\left[\text{sampled an available color}\right] \\
        & \geq \Pr\left[\text{$v$ is \textsf{awake}}\right] \cdot \Pr\left[\text{sampled an available color} \mid X < \frac{3}{4} d_v\right] \cdot \Pr\left[X < \frac{3}{4} d_v\right] \\
        & \geq \frac{1}{2} \cdot \frac{\frac{1}{4} d_v + 1}{d_v + 1} \cdot \frac{1}{3} \\
        & \geq \frac{1}{24} .
    \end{align*}

    Thus, the success probability that vertex \( v \) is marked as \textsf{done} (i.e., not \textsf{active} in the next iteration) can be lower bounded as
    
    \[
    \Pr[A_{i+1} = 0 \mid A_{i} = 1] \geq \frac{1}{24} .\qedhere
    \]  
\end{proof}

Next, we will show that the upper bound of the probability that a vertex \( v \) is \textsf{active} in the \( i \)th iteration decreases exponentially with \( i \).

\begin{lemma}
    \label{lm:active2}
    \( \Pr[A_i = 1] \leq \left(\frac{23}{24}\right)^{i-1} \), where $i \in \mathbb{N}$.
\end{lemma}

\begin{proof}
The base case holds since each vertex is active in the first iteration, i.e., $\Pr[A_1=1]=1$.
    By the inductive hypothesis, for each $i \geq 2$, we have \( \Pr[A_{i-1} = 1] \leq \left(\frac{23}{24}\right)^{i-2} \). Moreover, by \Cref{lem:active}, it follows that 
    \[
    \Pr[A_{i} = 1 \mid A_{i-1} = 1] \leq \frac{23}{24} \quad \text{and} \quad \Pr[A_{i} = 1 \mid A_{i-1} = 0] = 0.
    \]
    Therefore, for \( i \geq 2 \):

    \begin{align*}
        \Pr[A_{i} = 1] & = \Pr[A_{i-1} = 1] \cdot \Pr[A_{i} = 1 \mid A_{i-1} = 1] + \Pr[A_{i-1} = 0] \cdot \Pr[A_{i} = 1 \mid A_{i-1} = 0] \\
        & \leq \left( \frac{23}{24} \right)^{i-2} \cdot \frac{23}{24} + 0 \\
        & = \left(\frac{23}{24}\right)^{i-1}.\qedhere
    \end{align*}

\end{proof}

\begin{lemma}
    \label{lm:active3}
    The expected number of \textsf{active} vertices in the \( i \)th iteration is at most \( n \cdot \left( \frac{23}{24} \right)^{i-1} \).
\end{lemma}

\begin{proof}
    The proof follows from applying \Cref{lm:active2} and using the linearity of expectation.
\end{proof}
Now, we prove the first claim in \Cref{lm:main}.

\begin{proof}[Proof of \Cref{lm:main}(i)]
   Note that the number of iterations in \randcolortrial is \comments{$\lceil 1+4\log_{24/23} \log n \rceil$}. Hence, by \Cref{lm:active3}, the expected number of \textsf{active} vertices is ${n}/{\log^4 n}$ at the end of the procedure.
\end{proof}
 %We are done with the proof of the first claim in \Cref{lm:main}.

\subsection{Communication and round complexity of random color trial}\label{sec:comm}
To prove the expected communication complexity of \randcolortrial, we prove that the expected communication cost of the procedure is upper bounded by a constant \emph{for each vertex} $v$, and then we apply linearity of expectation.

For a vertex $v$ \textsf{active} in the $i$th %\yijun{I would prefer $i$th over $i$-th, but I am ok with both, as long as we do this consistently everywhere} 
iteration of $\mathsf{Random\mbox{-}Color\mbox{-}Trial}$, the communication is due to executing $\mathsf{Color\mbox{-}Samp}$ to sample a color $c$ from the set of available colors, and then due to confirming whether $c$ is used neither by any vertex in $N_A(v)$ nor in $N_B(v)$, so that $v$ can be colored with color $c$. We refer to the communication cost incurred due to the call of $\mathsf{Color\mbox{-}Samp}$ (over all iterations of $\mathsf{Random\mbox{-}Color\mbox{-}Trial}$) as \emph{Color-Sampling complexity}. Moreover, the communication cost incurred (over all iterations of $\mathsf{Random\mbox{-}Color\mbox{-}Trial}$) to verify whether the sampled color using $\mathsf{Color\mbox{-}Samp}$ is valid is called \emph{Color-Confirmation complexity}. We separately show that the expected value of both of them is $O(1)$.

%Analyzing bit complexity is one of the most challenging aspects of vertex coloring in parallel. In the previous work by \citet{10.1145/3662158.3662796}, it was shown that the expected communication is \( O(n) \) bits in the sequential algorithm. To minimize the round complexity, we use the Random Color Trial technique to parallelize the algorithm, which complicates the analysis. O

\paragraph{Color-confirmation complexity:}  In \Cref{alg:RandomColorTrial}, if vertex $v$ is \textsf{active} in the $i$th iteration, the two parties exchange a one-bit confirmation to verify whether the sampled color is used by any vertex in $N_A(v)$ or in $N_B(v)$. Therefore, the expected  Color-Confirmation complexity for vertex  \( v \) is the expected number of iterations in which \( v \) remains \textsf{active}. %The following lemma argues that the number of iterations a vertex is \textsf{active} is $O(1)$ in expectation. 

Recall from \Cref{sec:active} that $A_i$ is the indicator random variable such that \( A_{i} = 1 \) if the vertex \( v \) is \textsf{active} in the \( i \)th iteration.

\begin{lemma}
    \label{lm:active_iteration}
    The expected number of iterations that a vertex \( v \) is \textsf{active} is constant. That is, the Color-Confirmation cost of a vertex $v$ is $O(1)$ in expectation.
\end{lemma}

\begin{proof}
Let \( A \) denote the number of iterations where vertex \( v \) is \textsf{active}. Therefore, $A=\sum_{i=1}^{\comments{\lceil 1+4\log_{24/23} \log n \rceil}}A_i.$
 By \Cref{lm:active2}, $\E[A_i]=\Pr[A_i=1]\leq (23/24)^{i-1}$. By linearity of expectation,   we have
    \[
    \E[A] =  \sum_{i=1}^{\comments{\lceil 1+4\log_{24/23} \log n \rceil}} \E[A_i] ~~ \leq~~ \sum_{i=1}^{\infty} \left(\frac{23}{24}\right)^{i-1} =~~ O(1).\qedhere
    \]
\end{proof}

Hence, from the above lemma, the Color-Confirmation complexity for a vertex $v$ is $O(1)$ in expectation.

%By \Cref{lm:active_iteration}, we can imply that the expected number of bits of communication for each vertex in this confirmation process is \( O(1) \).

\paragraph{Color-sampling complexity:}
Let \( C_i \) be the random variable representing the communication cost for sampling a color using \colorsamp for vertex \( v \) in the \( i \)th iteration.  We can upper bound the expected value of $C_i$ as follows.

\begin{lemma}\label{lem:comm-i}
    $\E[C_i]\leq \left(\frac{23}{24}\right)^{i-1} \cdot \E[C_i \mid A_i = 1],$ where $i \in \mathbb{N}$.
\end{lemma}
\begin{proof} Note that $C_i=0$ if vertex $i$ not \textsf{active} in the $i$th iteration, i.e., there is no communication in the \( i \)th iteration if vertex \( v \) is \textsf{done} in some previous iteration. By using \Cref{lm:active2}, the probability that vertex $v$ is active in the $i$th iteration is $(23/24)^{i-1}$. Hence,  
    \begin{align*}
    \E[C_i] &= \Pr[A_i = 1] \cdot \E[C_i \mid A_i = 1] + \Pr[A_i = 0] \cdot \E[C_i \mid A_i = 0] \\
    &= \Pr[A_i = 1] \cdot \E[C_i \mid A_i = 1] \\
    &\leq \left(\frac{23}{24}\right)^{i-1} \cdot \E[C_i \mid A_i = 1]. \qedhere
\end{align*}
\end{proof}

Let \( C \) be the total Color-Sampling complexity for vertex $v$.  Therefore, $C=\sum_{i=1}^{\lceil 1+4\log_{24/23} \log n \rceil} C_i.$
 We divide the analysis to bound the expected value of $C$ into expected costs in two phases: before and after the \( j \)th iteration, where \( \comments{j = \lceil 1+2\log_{24/23} \log \Delta \rceil} \).  That is,

\begin{equation}\label{eqn:color-samp}
\E[C] = \sum_{i = 1}^j \E[C_i] + \sum_{i = j + 1}^{\comments{\lceil 1+4\log_{24/23} \log n \rceil}} \E[C_i].
\end{equation} 

The choice of $j$ is based on the fact that the probability of vertex \( v \) being \textsf{active} after the $j$th iteration is at most \( {1}/{\log^2 \Delta} \) by \Cref{lm:active2} and  the worst-case communication complexity of \colorsamp is \( O(\log^2 \Delta) \) by \Cref{lm:ksi}.
In the next lemma, we argue that the total Color-Sampling complexity for vertex $v$ after $j$th iteration, i.e, $\sum_{i = j + 1}^{\comments{\lceil 1+4\log_{24/23} \log n \rceil}} \E[C_i]$ is $O(1)$ in expectation. In the subsequent lemma, we prove a similar guarantee about $\sum_{i = 1}^j \E[C_i] $, i.e.,  the total Color-Sampling complexity for vertex $v$ before the $j$th iteration.

%Consider the expected communication cost after the \( j \)th iteration for vertex \( v \). Note that the expected communication cost in each iteration is upper bounded by \( O(\log^2 \Delta) \).

\begin{lemma}\label{lm:after}
For $j=\comments{\lceil 1+2\log_{24/23} \log \Delta \rceil}$, $\sum\limits_{i = j + 1}^{\comments{\lceil 1+4\log_{24/23} \log n \rceil}} \E[C_i]=O(1)$.
\end{lemma}
\begin{proof}
By \Cref{lem:comm-i}, $\E[C_i]\leq \left(\frac{23}{24}\right)^{i-1} \cdot \E[C_i \mid A_i = 1]$ for any $i \in \mathbb{N}$. Since the worst case communication complexity of \colorsamp is $O(\log^2 \Delta)$ by \Cref{lm:ksi}, we have $\E[C_i \mid A_i = 1]=O(\log^2 \Delta)$.  Hence, 

\begin{align*}
    \sum_{i = j + 1}^{\comments{\lceil 1+4\log_{24/23} \log n \rceil}} \E[C_i] &= \sum_{i = j + 1}^{\comments{\lceil 1+4\log_{24/23} \log n \rceil}} \left(\frac{23}{24}\right)^{i-1} \cdot \E[C_i \mid A_i = 1] \\
    &\leq \sum_{i = j + 1}^{\infty} \left(\frac{23}{24}\right)^{i-1} \cdot O(\log^2 \Delta) \\
    &\leq O(\log^2 \Delta) \cdot \sum_{i = j + 1}^{\infty} \left(\frac{23}{24}\right)^{i-1} \\
    &\leq O(\log^2 \Delta) \cdot 23 \cdot \left(\frac{23}{24}\right)^{j-1} \\
    &\leq O\left(\log^2 \Delta\right) \cdot \frac{1}{\log^2 \Delta} = O(1).\qedhere
\end{align*}
\end{proof}

From now, we mainly focus on proving the following lemma, hence bounding expected Color-Sampling cost  prior to the \( j \)th iteration.

\begin{lemma}
    \label{lm:bit}
    For \( \comments{j = \lceil 1+2\log_{24/23} \log \Delta \rceil }\), it follows that
    $
    \sum\limits_{i = 1}^j \E[C_i] = O(1).
    $
\end{lemma}

Consider vertices with \emph{initial} low degree in the original graph. If a vertex \( v \) has a sum of degrees in both Alice's and Bob's graphs that is less than \( \Delta/2 \) at the outset, we can demonstrate that its expected bit complexity remains constant. This follows from the fact that the number of available colors in $v$'s palette is at least $\Delta/2$, i.e., \( k_v \geq \Delta/2 \), since its neighbors can occupy at most half of the colors. Consequently, during each iteration in which vertex \( v \) is \textsf{active} and \textsf{awake}, it utilizes \( O(\log^2((\Delta+1/k_v)) = O(1) \) bits in expectation. Applying the same analysis as above, we obtain

\begin{align*}
\sum_{i = 1}^j \E[C_i] & \leq \sum_{i = 1}^j \left(\frac{23}{24}\right)^{i-1} \cdot \E[C_i \mid A_i = 1] \\
& \leq \sum_{i = 1}^j \left(\frac{23}{24}\right)^{i-1} \cdot O(1) \\
&\leq O(1) \cdot \sum_{i = 1}^\infty \left(\frac{23}{24}\right)^{i-1} \\
& = O(1) \cdot 24 = O(1).
\end{align*}

\begin{remark}
    It is important to note that Alice and Bob do not have knowledge of the classification of the initial low-degree vertices and therefore do not treat these vertices differently.
\end{remark}

We now turn our attention to vertices with \emph{initial} high degree in the original graph, where the combined degree of both parties is at least \( \Delta/2 \). A vertex is classified as \textsf{high-degree} if, in the \(i\)th iteration, it is \textsf{active} and has at least \( {\Delta}/{2^{2i}} \) \textsf{active} neighbors. If it does not meet this criterion, the \textsf{active} vertex is classified as \textsf{low-degree}. In the first iteration, all considered vertices are \textsf{high-degree}. This classification is significant because it is generally more challenging to assign an available color to \textsf{low-degree} vertices, as many of their neighbors may have already occupied colors.

The probability of a vertex being \textsf{low-degree} is quite low. During an iteration, for a vertex to be marked as \textsf{low-degree}, more than three-fourths of its \textsf{active} neighbors must be \textsf{awake} and obtained a color. We can quantify this low probability using the Chernoff bound.

\begin{proposition}[Chernoff Bound] Let \( X_1, X_2, \ldots, X_n \) be independent random variables in \([0, 1]\). Let \( X = \sum_{i=1}^n X_i \) and $\mu = \E[X]$. Then, for any \(0 < \delta < 1\),
\[
\Pr[X \leq (1 - \delta)\mu] \leq e^{-\frac{\delta^2 \mu}{2}}.
\]
\end{proposition}

\begin{lemma}
    In the $i$th iteration, the probability that a \textsf{high-degree} vertex becomes \textsf{low-degree} is at most $e^{-\frac{\Delta}{2^{2i+4}}}$.
\end{lemma}

\begin{proof}
    In the $i$th iteration, we bound the probability that fewer than a quarter of a vertex's \textsf{active} neighbors are \textsf{idle} using Chernoff bounds. A \textsf{high-degree} vertex has at least $\frac{\Delta}{2^{2i}}$ \textsf{active} neighbors, and each \textsf{idle} with probability $\frac{1}{2}$ independently. Let $X$ be the number of \textsf{active} neighbors \textsf{idle} in the next round. We have $\E[X] \geq \frac{\Delta}{2^{2i+1}}$. By Chernoff bounds with $\delta = 1/2$, the probability that fewer than a quarter of the neighbors \textsf{idle} is:
    \[
    \Pr\left[X \leq \left(1 - \frac{1}{2}\right) \frac{\Delta}{2^{2i + 1}}\right] \leq \Pr[X \leq (1 - \delta)\mu] \leq e^{-\frac{\delta^2 \mu}{2}} \leq e^{-\frac{\Delta}{2^{2i+4}}}.
    \]
    If more than a quarter of the \textsf{active} neighbors are \textsf{idle}, the vertex cannot become \textsf{low-degree}. Hence, the probability that a \textsf{high-degree} vertex becomes \textsf{low-degree} in the $i$th iteration is at most $e^{-\frac{\Delta}{2^{2i+4}}}$.
\end{proof}

% \begin{lemma}
%     For each vertex, the probability that it is a \textsf{low-degree} vertex in the $i$th iteration is at most $e^{-\frac{\Delta}{2^{2i+4}}}$.
% \end{lemma}

\comments{Let us upper bound \( e^{-\frac{\Delta}{2^{2i+4}}} \). Note that the function \( e^{-\frac{\Delta}{2^{2i'+4}}} \) is strictly increasing with respect to \( i \). Thus, we have \( e^{-\frac{\Delta}{2^{2i'+4}}} > e^{-\frac{\Delta}{2^{2i+4}}} \) for all \( i' > i \). This is why it is important to bound \( j \) by \( O(\log \log \Delta) \). Specifically, $e^{-\frac{\Delta}{2^{2j+4}}}=  e^{-\frac{\Delta}{2^{O(\log \log \Delta)}}} \leq {1}/{{\log^4 \Delta}} $ for sufficiently large \( \Delta \). Hence, for all $i \leq j$, $e^{-\frac{\Delta}{2^{2i+4}}}\leq {1}/{{\log^4 \Delta}}$.} In the preceding claim, we assume that \( \Delta \) exceeds a sufficiently large constant. This assumption is valid because, if \( \Delta \) is a constant, we can simply use the worst-case analysis, which has an expected bit complexity of \( O(\log^2 \Delta) = O(1) \).

Next, we seek to upper bound the probability that a vertex becomes \textsf{low-degree} in the \( i \)th iteration.

\begin{lemma}
    \label{lm:low}
Consider a \textsf{high-degree} node in iteration $(i-1)$th, The probability that a vertex becomes \textsf{low-degree} in the \( i \)th iteration for some \( i \leq j \) is at most \( {1}/{{\log^3 \Delta}} \), where \( j = \lceil 1 + 2\log_{24/23} \log \Delta \rceil \).

\end{lemma}

\begin{proof}
    By the \emph{union bound}, the probability that a vertex becomes \textsf{low-degree} can be upper bounded as follows:%\yijun{Note: should use QEDHERE to place the qed-symbol in the right place if needed, like below.}
    \[
    \sum_{i=1}^{j} e^{-\frac{\Delta}{2^{2t+4}}} \leq \sum_{i=1}^{j} \frac{1}{\log^3 \Delta} \leq (\lceil 1+2\log_{24/23} \log \Delta \rceil) \cdot {\frac{1}{\log^4 \Delta}}  \leq \frac{1}{{\log^3 \Delta}}. \qedhere
    \]
    Note that we have assumed that $\Delta$ is sufficiently large.
\end{proof}
\color{black}
 Now, we prove \Cref{lm:bit}.

\begin{proof}[Proof of \Cref{lm:bit}]
When vertex $v$ has \emph{initial} high degree in the original graph, i.e.,  the combined degree of both parties is at least \( \Delta/2 \). That is, we show that $\sum_{i = 1}^j \E[C_i]=O(1)$.

%\begin{proof}[Proof of \Cref{lm:bit}]% By \Cref{lem:comm-i},
Again using the fact that $C_i=0$ if vertex $i$ not \textsf{active} in the $i$th iteration,
    \[
        \sum_{i = 1}^j \E[C_i] = \sum_{i = 1}^j \Pr[A_i = 1] \cdot \E[C_i \mid A_i = 1].
    \]
    
    Let \( H_i \) and \( L_i \) denote the events that vertex \( v \) is \textsf{high-degree} and \textsf{low-degree}, respectively, in the \( i \)th iteration.
    
    \[
        \sum_{i = 1}^j \E[C_i] = \sum_{i = 1}^j \Pr[H_i = 1] \cdot \E[C_i \mid H_i = 1] + \Pr[L_i = 1] \cdot \E[C_i \mid L_i = 1].
    \]
    
    Recall from \Cref{lm:ksi} that if the palette size is \( k \), then \Cref{alg:k-Slack-Int} samples a color using \( O\left(\log^2 ((\Delta+1)/k)\right) \) bits of communication in expectation. Thus, the expected cost for sampling a color for a \textsf{high-degree} vertex in the \( i \)th iteration is $\E[C_i \mid H_i = 1] = O\left(\log^2\left((\Delta+1)/\frac{\Delta}{2^{2i}}\right)\right) = O(i^2).$ For \textsf{low-degree} vertices, the worst-case expected bit complexity is $\E[C_i \mid L_i = 1] \leq O(\log^2 \Delta).$
    
    Therefore, we have
    
    \begin{align*}
        \sum_{i = 1}^j \E[C_i] & = \sum_{i = 1}^j \Pr[H_i = 1] \cdot \E[C_i \mid H_i = 1] + \Pr[L_i = 1] \cdot \E[C_i \mid L_i = 1] \\
        & \leq \sum_{i = 1}^j \Pr[H_i = 1] \cdot O(i^2) + \sum_{i = 1}^j \Pr[L_i = 1] \cdot O(\log^2 \Delta).
    \end{align*}
    
    In the \( i \)th iteration, a vertex is classified as \textsf{high-degree} only if it is \textsf{active}. According to \Cref{lem:active}, we have $\Pr[H_i = 1] \leq \Pr[A_i = 1] \leq \left(23/24\right)^{i-1}.$ \comments{The probability that a vertex becomes \textsf{low-degree} in iteration \( i \) for some \( i \leq j \) is at most \( {1}/{{\log^3 \Delta}} \), as shown in \Cref{lm:low}.} {Thus, $\Pr[L_i=1]\leq {1}/{\log^3{\Delta}}$ for each $i \leq j$.}

    Thus, we conclude that
    
    \[
       { \sum_{i = 1}^j \E[C_i] \leq \sum_{i = 1}^j \left(\frac{23}{24}\right)^{i-1} \cdot O(i^2) +  \lceil 1 + 2\log_{24/23} \Delta \rceil \cdot \frac{1}{\log^3 \Delta} \cdot O(\log^2 \Delta) = O(1).}    \]
    
%\end{proof}

Hence, we are done with the proof of \Cref{lm:bit} when vertex $v$ has initial high degree. %\yijun{Similarly, can we put the proof of \Cref{lm:bit} into proof environment? The structure of the lemmas/proofs is a bit messy without putting proofs into proof environment properly} 
Earlier, we showed the same when vertex $v$ has initial low degree.     
\end{proof}

By \Cref{lm:bit} and \Cref{lm:after} along with \Cref{eqn:color-samp}, the Color-Sampling cost in $\mathsf{Random\mbox{-}Color\mbox{-}Trial}$ for a vertex $v$ is $O(1)$ in expectation.

\begin{lemma}\label{lem:color-samp}
    For a vertex $v$, the total Color-Sampling cost in \randcolortrial  is $O(1)$ in expectation.
\end{lemma}
Now, we are ready to prove the second claim of \Cref{lm:main} which is about the total communication complexity of \randcolortrial.
\begin{proof}[Proof of \Cref{lm:main} (ii)]
Recall that the total Color-Confirmation cost for a vertex $v$ is $O(1)$ in expectation due to \Cref{lm:active_iteration}. Hence, combining \Cref{lm:active_iteration} with \Cref{lem:color-samp}, we have that the expected communication cost of $\mathsf{Random\mbox{-}Color\mbox{-}Trial}$ \emph{for each vertex} is upper bounded by a constant. By linearity of expectation, we are done with the proof of \Cref{lm:main} (ii), i.e., the expected communication complexity of $\mathsf{Random\mbox{-}Color\mbox{-}Trial}$ is $O(n)$ bits.
\end{proof}

Now, we analyze the round complexity of \randcolortrial and prove \Cref{lm:main} (iii).

\begin{proof}[Proof of \Cref{lm:main} (iii)]
    
In each iteration, colors are sampled for all \textsf{active} vertices in parallel. According to \Cref{lm:ksi}, the protocol requires \( O(\log \Delta) \) rounds to sample a color from the available palette in the worst case. Since this process is repeated for \( O(\log \log n) \) iterations, the overall round complexity of \randcolortrial is \( O(\log \log n \cdot \log \Delta) \) in the worst case.%\yijun{I think it would be better if there is a ``proof of \Cref{lm:main}'' so that the reader can easily locate the proof and see the dependencies between the lemmas. Same for other lemmas where the proof is not in the proof environment.}
\end{proof}

\subsection{Final algorithm and the proof of \Cref{thm:vertex}}\label{sec:final-vertex-col}

\mainvertex*
\begin{proof}

The protocol for the $(\Delta+1)$-vertex coloring problem  proceeds as follows:

\begin{center}
    \fbox{
        \begin{minipage}{0.9\textwidth}
            \begin{enumerate}
                \item \textbf{Run \randcolortrial}: Run \randcolortrial (\Cref{alg:RandomColorTrial}) on $G = (V, E)$. Let $Z$ be the set of uncolored vertices remaining at the end of the procedure.

                \item \textbf{Formulate \DLC instance}: Consider the subgraph $G_Z$ induced by the vertices in $Z$. For each vertex $v \in Z$, let $A_v$ and $B_v$ be the sets of colors assigned to already colored vertices in $N_A(v)$ and $N_B(v)$, respectively. Define the sets of available colors for $v$ as $\Psi_A(v) = [\Delta+1] \setminus A_v$ and $\Psi_B(v) = [\Delta+1] \setminus B_v$. Note that the graph $G_Z$ along with the sets $\Psi_A(v)$ and $\Psi_B(v)$ for $v \in Z$ forms a valid \DLC instance in the two-party communication model.

                \item \textbf{Solve \DLC}: Run the protocol for \DLC, as described in \Cref{lem:d1lc}, to color the vertices in $Z$.
            \end{enumerate}
        \end{minipage}
    }
\end{center}

Now, we analyze the above protocol.
\begin{description}
    \item[Correctness of the protocol:] Due to \Cref{lm:main}, \randcolortrial produces a valid partial coloring for all vertices in $V \setminus Z$. Furthermore, by \Cref{lem:d1lc}, the protocol for \DLC ensures a valid coloring for all the vertices in $Z$. Thus, the combined protocol results in a valid vertex coloring of the entire graph $G$.

    \item[Communication complexity of the protocol:] Let $C_1$ and $C_2$ denote the total communication costs incurred by \randcolortrial and the \DLC protocol, respectively. According to \Cref{lm:main} (ii), the expected communication complexity of \randcolortrial is $\mathbb{E}[C_1] = O(n)$ bits. To calculate $\mathbb{E}[C_2]$, note that $C_2$ depends on the random variable $Z$, which represents the number of uncolored vertices remaining after \randcolortrial. By \Cref{lem:d1lc} (i), we have $\mathbb{E}[C_2 | Z] = O(|Z| \cdot \log^2 |Z| \cdot \log^2 \Delta + |Z| \cdot \log^3 |Z|)$ bits. Since $|Z| \leq n$ and $\Delta \leq n$, it follows that $\mathbb{E}[C_2 | Z] = O(|Z| \cdot \log^4 n)$ bits. Furthermore, by \Cref{lm:main} (i), the expected number of vertices in $Z$ is $O(n / \log^4 n)$.
 Therefore,
    \[
    \mathbb{E}[C_2] = \mathbb{E}[\mathbb{E}[C_2 \mid Z]] = O\left( \mathbb{E}[|Z| \cdot \log^4 n] \right) = O(n).
    \]
    Hence, the expected communication complexity of the protocol  is $\mathbb{E}[C_1] + \mathbb{E}[C_2] = O(n)$ bits.

    \item[Round complexity of the protocol:] The round complexity of \randcolortrial is $O(\log \log n \cdot \log \Delta)$, due to \Cref{lm:main} (iii), and the round complexity of the \DLC protocol is $O(\log \Delta)$, due to \Cref{lem:d1lc}. Hence, the total round complexity of the protocol  is $O(\log \log n \cdot \log \Delta)$ in the worst case.\qedhere
\end{description}
\end{proof}

% Thus, putting things together, we conclude the proof of \Cref{thm:vertex}.
%\yijun{Maybe put the above discussion into a proof environment?}

\section{Protocol for \texorpdfstring{$(2\Delta-1)$}{(2Delta-1)}-edge coloring}\label{sec:edge-upper}
\label{sec:edge-protocol}
%\yijun{Separate the upper and lower bounds into two sections?}
% \gopinath{May be we need some introduction to the section.}
% \hung{Yes, I will add soon.}

In this section, we present a deterministic protocol for the $(2\Delta-1)$-edge coloring problem that uses $O(n)$ bits of communication and requires $O(1)$ rounds, thus proving \Cref{thm:edge}. In \Cref{sec:edge-existent}, we discussed two classical results in edge coloring that are used in our protocol.

\paragraph{Bounded degree graph:}
First, if \(\Delta\) is a constant, finding a \(2\Delta - 1\)-edge coloring is straightforward. We will demonstrate that there exists a simple deterministic protocol requiring \(O(n)\) bits of communication.

\begin{lemma}\label{lem:constant-degree-edge-col}
If \(\Delta\) is a constant, there exists a deterministic protocol that requires \(O(n)\) bits and one round of communication to achieve a \((2\Delta - 1)\)-edge coloring.
\end{lemma}

\begin{proof}
The case where \(\Delta = 1\) is trivial, as the two parties can simply use the only available color to color their respective edges.

For \(\Delta \geq 2\), Alice colors her edges using \(2\Delta - 1\) colors through a greedy coloring algorithm. Since each edge is adjacent to at most \(2\Delta - 2\) other edges, it is always possible to find a color for each edge. After coloring, she transmits the remaining available colors for each vertex to Bob. Because the number of colors is constant, this information can be communicated in \(O(n)\) bits within a single round. Bob can then apply a greedy coloring algorithm on his graph, as there will always be at least one available color for each edge in Bob's graph.
\end{proof}

% From this point onward, we assume that $\Delta$ is at least four.\yijun{At some appropriate place, make it clear where we use this assumption} Within the \(2\Delta - 1\) colors, Alice will use \(\Delta - 1\) colors, Bob will use another set of \(\Delta - 1\) colors, and the only remaining color is called the \textsf{special} color. Through \(O(n)\) bits of communication, Alice and Bob can exchange the classification of vertices in their respective graphs by sending an array of \(n\) bits. For instance, they can indicate the vertices with degree \(\Delta - 1\) by marking the corresponding entries in the array with bit \(1\). From this point onward, we assume that Alice and Bob are aware of the vertices with degree \(\Delta\) and the vertices with degree \(\Delta - 1\) in each other's graphs. We will now describe the protocol concerning Alice's graph, noting that when we refer to a vertex of degree \(\Delta\), we are indicating its degree in Alice's graph unless otherwise specified explicitly.\yijun{The above paragraph can be made more precise.}

\paragraph{Unbounded degree graph:}
Alice and Bob agree on a \emph{partition} of \(2\Delta - 1\) colors into three sets: two sets of size \(\Delta - 1\) each, known as Alice's palette and Bob's palette, and one \textsf{special} color. We now describe the protocol \emph{concerning Alice's graph}, noting that when we refer to the degrees of vertices, we specifically mean their degree in Alice's graph unless stated otherwise. The protocol follows the same structure on Bob's side. See \Cref{alg:EdgeColoringProtocol}.

\begin{algorithm}[ht!]%<- this is to avoid large white space
\caption{Deterministic \(2 \Delta - 1\) Edge Coloring Protocol with $\Delta \geq 8$}
\label{alg:EdgeColoringProtocol}
\KwIn{Given a graph \(G = (V, E)\) with maximum degree \(\Delta\), where \(E\) is partitioned adversarially into \(E_A\) and \(E_B\). Alice receives \(G_A = (V, E_A)\) and Bob receives \(G_B = (V, E_B)\).}
\KwOut{A proper edge coloring of the graphs using \(2\Delta - 1\) colors.}

\tcp{Alice's protocol (the same protocol applies to Bob)}
$\textsc{DG} \gets (V, \emptyset)$ \tcp{Deferred subgraph}

$\textsc{RG} \gets (V, E_A)$ \tcp{Remaining subgraph}

\While{there exists an edge \(e \in \textsc{RG}\) connecting vertices of degree at least \(\Delta - 1\)}{
    Add \(e\) to \textsc{DG} and remove \(e\) from \textsc{RG}.
}

Find a \(\Delta\)-perfect matching \(M\) by \Cref{lm:matching}.

\For{edge \(e \in M\)}{
    Remove \(e\) from \textsc{RG};
}

Find a \((\Delta - 1)\)-edge coloring for the remaining subgraph \textsc{RG} using Alice's palette by \Cref{thm:independent}.

Alice and Bob send two arrays of \(n\) bits to indicate vertices of degree larger than \(\Delta / 2\), and vertices covered by their respective \(\Delta\)-perfect matching.

\For{edge \(e \in M\) that not incident to any vertex covered by Bob's \(\Delta\)-perfect matching}{
    Color $e$ with the \textsf{special} color.
}

\For{edge \(e \in M\) that incident to vertex $v$ covered by Bob's \(\Delta\)-perfect matchings}{
    \If{degree of \(v\) in Bob's graph is greater than \(\Delta / 2\)}{
        Color $e$ with the \textsf{special} color.
    } \Else{
        Find a color from Bob’s available palette of vertex $v$ using \Cref{lm:edge-sample} to color the edge $e$. \label{step:edge-sample}
    }        
}

Find a proper coloring of the deferred subgraph \textsc{DG} using the first seven colors from Bob's palette by \Cref{lm:defer}.
\end{algorithm}

%\gopinath{May be we can highlight (in the pseudocode) that the protocol is for $\Delta \geq 4.$}\hung{Yes}
%\yijun{In the pseudocode why do we need 4 cases for match vertices and not just 2 cases (either use special color or find a replacement color)}\hung{I have done with the simplification}

% They exchange an array of \(n\) bits, where each entry with value \(1\) indicates that the corresponding vertex has degree \(\Delta\) in their respective graphs. Additionally, they exchange another array of \(n\) bits to indicate which vertices have degree \(\Delta - 1\). 

\paragraph{Deferring edges:}
Recall \Cref{thm:independent}: if the vertices of maximum degree \( \Delta - 1 \) form an independent set, then there exists a proper edge-coloring with \( \Delta - 1 \) colors. To achieve this, we defer the coloring of several edges such that the remaining subgraph satisfies the conditions of the theorem. We denote the deferred subgraph as \( \textsc{DG} = (V, \emptyset) \), which initially contains no edges, and the remaining subgraph as \( \textsc{RG} = (V, E_A) \), where \( E_A \) represents the edges initially present in Alice's part. Alice then sequentially identifies edges in the subgraph \( \textsc{RG} \) that connect vertices with the following pairs of degrees: 
\[
\{ \{\Delta, \Delta\}, \{\Delta, \Delta - 1\}, \{\Delta - 1, \Delta - 1\} \}.
\]
For each such edge, she adds it to the deferred subgraph \( \textsc{DG} \) and removes it from the remaining subgraph \( \textsc{RG} \). Alice continues this process until the remaining subgraph \( \textsc{RG} \) contains no edges that connect vertices of degree at least \( \Delta - 1 \).

%\gopinath{It seems that we want to add one edge at a time to DG. Then see a potential edge to be added to \textsc{DG} in graph $G$ after removing already added edges to DG?} \hung{Yes, how should I fix?}\gopinath{This looks fine now.}

\begin{lemma}
\label{lm:edge-defer}
The maximum degree of \textsc{DG} is \(2\).
\end{lemma}

\begin{proof}
Each vertex has a maximum degree of \(\Delta\), and an edge \emph{cannot} be added if it is incident to a vertex with a degree of less than or equal to \(\Delta - 2\). Consequently, each vertex can have at most two edges added to the deferred subgraph, implying that the maximum degree in the deferred graph \textsc{DG} is \(2\).
\end{proof}

\paragraph{\(\Delta\)-perfect matching:} The goal is for Alice to use her set of \(\Delta - 1\) colors to color the remaining subgraph \textsc{RG}. To achieve this, we must ensure that the subgraph has a maximum degree of \(\Delta - 1\) and that no vertices of degree \(\Delta - 1\) are connected, as stated in \Cref{thm:independent}. It is important to note that we apply \Cref{thm:independent} with respect to \((\Delta - 1)\)-edge coloring. Consider the current subgraph \textsc{RG}, where the vertices of degrees \(\Delta\) and \(\Delta - 1\) form an independent set. To address the vertices of degree \(\Delta\), we will remove one edge incident to each of these vertices from \textsc{RG} to satisfy the required conditions. To facilitate this, we will first find a matching—a set of disjoint edges—that covers every vertex of degree \(\Delta\). This is referred to as a \(\Delta\)-perfect matching.%, which exists due to Hall's marriage theorem.

%\yijun{emphasize that your goal is to apply this tool with max degree $\Delta-1$ and not $\Delta$ so that Alice can use her set of colors.}\hung{Done}  %\farrel{Add Hall's Marriage to overview or preliminary?}\yijun{Can do this, or can refer to it in the proof below.}

% \begin{proposition}[\citet{hall1987representatives}]
% Consider a bipartite graph $G = (X, Y, E)$ with bipartite sets \( X \) and \( Y \) and edge set \( E \). For a subset \( W \) of \( X \), let \( N_G(W) \) denote the neighborhood of \( W \) in \( G \), the set of all vertices in \( Y \) that are \emph{adjacent} to at least one element of \( W \). An $X$-perfect matching is a matching, a set of disjoint edges, which covers every vertex in $X$. There is an \( X \)-perfect matching \emph{if and only if} for every subset \( W \) of \( X \):
% \[
% |W| \leq |N_G(W)|
% \]
% \end{proposition}

\begin{lemma}
\label{lm:matching}
    In a graph with maximum degree $\Delta$ where vertices of degree $\Delta$ form an independent set, there exists a $\Delta$-perfect matching.
\end{lemma}
%\farrel{This doesn't work for every graph with max degree $\Delta$. Need to mention in the lemma that the graph also requires the vertices with degree $\Delta$ to be independent set?}\yijun{Agreed.}
\begin{proof}
Consider a bipartite graph \( G = (D, Y, E) \), where \( D \) consists of all vertices of degree \( \Delta \) in one partition, \( Y \) contains the remaining vertices, and \( E \) represents the edges that were originally incident to vertices in \( D \). Since the vertices in \( D \) form an independent set in the original graph, it suffices to show that there exists a matching in {graph} \( G \) that covers all vertices in \( D \).

  Let us consider the \emph{linear program} for the \emph{maximum fractional matching} of the above bipartite graph \( G \).

    \[
\begin{aligned}
    \text{maximize} \quad & \sum_{e \in E} x_e \\
    \text{subject to} \quad & \sum_{e \in N_G(v)} x_e \leq 1, \quad \forall v \in D \cup Y, \\
    & 0 \leq x_e \leq 1, \quad \forall e \in E.
\end{aligned}
\]

Observe that \( x_e = \frac{1}{\Delta} \) for each \( e \in E \) is a feasible solution to the above linear program. Recall that the degree of each vertex in \( D \) is \( \Delta \). Thus,  
\[
\sum_{e \in N_G(v)} x_e = 1, \quad \text{for each } v \in D,
\]
i.e., the fractional matching (with \( x_e = \frac{1}{\Delta} \) for all \( e \in E \)) covers all vertices in \( D \).

As the graph under consideration is a bipartite graph, and due to the fact that there is always an integral matching corresponding to a fractional matching in a bipartite graph \cite{schrijver2003combinatorial}, there exists an integral matching in \( G \) such that all vertices in \( D \) are covered.    
\end{proof}
% \begin{proof}
% Consider a bipartite graph \(G = (D, Y, E)\), where \(D\) is the set of all vertices of degree \(\Delta\) in one partition, \(Y\) is the set of remaining vertices, and \(E\) is the set of all edges that are \emph{originally} incident to vertices in \(D\). We will prove that \(|W| \leq |N_G(W)|\) for every \(W\) by contradiction. 

% Suppose there exists a subset \(W \subseteq D\) such that \(|W| > |N_G(W)|\). We will examine the number of edges in the sub-bipartite graph created by \(W\) and \(N_G(W)\), which we denote as \(E_W\). 

% First, since each vertex in \(W\) has degree \(\Delta\), we have \(E_W = \Delta \cdot |W|\). 

% Second, since the maximum degree in \(N_G(W)\) is \(\Delta - 1\), we can upper bound \(E_W\) by \((\Delta - 1) \cdot |N_G(W)|\). This leads to a contradiction, as it implies

% \[
% \Delta \cdot |W| \leq (\Delta - 1) \cdot |N_G(W)|,
% \]

% which contradicts our assumption that \(|W| > |N_G(W)|\).
% \end{proof}

Thus, each party can find \(\Delta\)-perfect matchings in their respective graphs. By removing the edges in the \(\Delta\)-perfect matchings from the remaining subgraph \textsc{RG}, we obtain a new subgraph where the maximum degree is \(\Delta - 1\) and the set of vertices with degree \(\Delta - 1\) forms an independent set. According to \Cref{thm:independent}, there exists a proper \((\Delta - 1)\)-edge coloring in the remaining subgraphs of Alice, allowing her to use her respective sets of \(\Delta - 1\) colors to color the edges. Next, we need to address the coloring of the deferred subgraph \textsc{DG} and the edges of the matching.

Notably, up to this point, no communication has occurred. Alice aims to color the edges in the \(\Delta\)-perfect matchings using either the \textsf{special} color or colors from Bob's palette. Both parties exchange an array of \(n\) bits to indicate which vertices are covered by their \(\Delta\)-perfect matchings.

Consider an edge \(\{u, v\}\) in Alice's \(\Delta\)-perfect matching, where \(u\) is a vertex of degree \(\Delta\) in Alice's graph. If \(v\) is not covered by Bob's \(\Delta\)-perfect matching, then Alice can color the edge \(\{u, v\}\) with the \textsf{special} color. Otherwise, there exists a vertex \(t\) of degree \(\Delta\) in Bob's graph such that \(\{t, v\}\) is in Bob's \(\Delta\)-perfect matching.

We define vertices that appear in both \(\Delta\)-perfect matchings as \textsf{match} vertices. Alice and Bob exchange another array of \(n\) bits to indicate which vertices have a degree greater than \(\Delta / 2\) in their graphs. For each \textsf{match} vertex \(v\) that is incident to an edge \(\{u, v\}\) in Alice's \(\Delta\)-perfect matching, if the degree of vertex \(v\) in Bob's graph is greater than \(\Delta / 2\), Alice will color the edge \(\{u, v\}\) with the \textsf{special} color. Otherwise, she will select an available color from Bob's palette for vertex \(v\).

It is important to note that there may be several edges incident to \(v\) in Bob's remaining subgraph that use colors from Bob's palette. However, it suffices to use \emph{an available color in Bob's palette for vertex \(v\)}, as \(u\) is a vertex of degree \(\Delta\) in Alice's remaining subgraph, making the edge \(\{u, v\}\) the only possible edge incident to \(u\) using Bob's palette.

\begin{restatable}{lemma}{edgesample}\label{lm:edge-sample}
There exists a deterministic protocol such that, for the set of vertices \(V_{\leq \Delta/2}^{(B)} = \left\{v \in V \mid \deg_{G_B}(v) \leq \Delta/2 \right\}\)
in Bob's graph \( G_B \), the protocol allows Alice to find one available color from Bob's palette for each vertex in \( V_{\leq \Delta/2}^{(B)} \), and vice versa. This is achieved using \( O(n) \) bits of communication and one rounds.
\end{restatable}

\begin{proof}
Consider a vertex \(v \in V_{\leq \Delta/2}^{(B)}\). Bob's palette contains \( \Delta - 1 \) colors. The maximum number of colors already used by the edges incident to \(v\) is at most \( \lfloor \Delta / 2 \rfloor \). Therefore, the number of colors available for \(v\) in Bob's palette is at least \( (\Delta - 1) - \lfloor \Delta / 2 \rfloor \).

The fraction of available colors for \(v\) is thus at least
\[
\frac{(\Delta - 1) - \lfloor \Delta / 2 \rfloor}{\Delta - 1}.
\]

For \( \Delta \geq 4 \), the fraction of available colors is at least \( 1/3 \). Our \Cref{alg:EdgeColoringProtocol} assumes \( \Delta \geq 8 \), so this condition holds.

Now, consider a subset of vertices \( K \subseteq V_{\leq \Delta/2}^{(B)} \) with \( |K| = k \). For each vertex \( v \in K \), at least a \( 1/3 \) fraction of Bob's \( \Delta - 1 \) colors are available in its palette. Let \( A_v \) be the set of available colors for vertex \( v \). We have \( |A_v| \geq \frac{\Delta - 1}{3} \).

We want to show that there exists a color \( c \) that is available in the palettes of at least a \( 1/3 \) fraction of the vertices in \( K \). We use a double-counting argument. Let \( S = \{ (v, c) \mid v \in K, c \in A_v \} \) be the set of pairs where color \( c \) is available for vertex \( v \).

We can count the size of \( S \) in two ways:
\begin{enumerate}
    \item By summing over the vertices: For each vertex \( v \in K \), there are \( |A_v| \geq \frac{\Delta - 1}{3} \) available colors. Thus, \( |S| = \sum_{v \in K} |A_v| \geq k \cdot \frac{\Delta - 1}{3} \).
    \item By summing over the colors: Let \( n_c \) be the number of vertices in \( K \) for which color \( c \) is available. Then, \( |S| = \sum_{c \in \text{Bob's Palette}} n_c \).
\end{enumerate}

Suppose, for the sake of contradiction, that for every color \( c \), \( n_c < \frac{k}{3} \). Since there are \( \Delta - 1 \) colors in Bob's palette, we would have:
\[
|S| = \sum_{c \in \text{Bob's Palette}} n_c < (\Delta - 1) \cdot \frac{k}{3}.
\]
However, we also know that \( |S| \geq k \cdot \frac{\Delta - 1}{3} \), which leads to a contradiction. Therefore, there exists at least one color \( c_K \) such that it is available in the palettes of at least \( \frac{1}{3} \) fraction of the vertices in \( K \).

Now, Bob iteratively finds a set of colors that covers the available colors in the palettes of all vertices in \( V_{\leq \Delta/2}^{(B)} \). Let \( U_1 = V_{\leq \Delta/2}^{(B)} \). In the \( i \)th step, we find a color \( c_i \) that is available for at least \( 1/3 \) fraction of the vertices in the current set \( U_{i} \). Let \( V_i \) be the binary array indicating which vertices in \( U_{i} \) with color \( c_i \) available. We then define \( U_{i+1} \) as the subset of vertices in \( U_{i} \) for which color \( c_i \) is not available (corresponding to the $0$s in \( V_i \)).

The size of \( U_{i+1} \) is at most \(({2}/{3}) |U_{i}| \). Starting with \( |U_0| = |V_{\leq \Delta/2}^{(B)}| \leq n \), after \( t = O(\log n)\) iterations, we have \( |U_t| \leq (2/3)^t n < 1\), which means all vertices have at least one of the chosen colors available. Thus, there exists a set of \( O(\log n) \) colors \( \{c_1, \ldots, c_{O(\log n)}\} \) such that every vertex in \( V_{\leq \Delta/2}^{(B)} \) has at least one of these colors available in its palette.

Finally, Bob sends this set of \( O(\log n) \) colors. Along with this, Bob sends a sequence of binary arrays \( V_1, V_2, \ldots, V_{O(\log n)} \). 

The total number of bits sent in these binary arrays is the sum of their lengths:
\[
\sum_{i=1}^{O(\log n)} \text{length}(V_i) = \sum_{i=1}^{O(\log n)} |U_i| \leq \sum_{j=0}^{\infty} \left(\frac{2}{3}\right)^i \cdot |U_1| \leq 3|U_1| = O(n)
\]

Therefore, the total communication cost for this part is \( O(\log n) \cdot O(\log \Delta)\) for the colors and \( O(n) \) for the binary arrays, resulting in an overall communication of \( O(n) \) bits and one round.
\end{proof}

% We prove \Cref{lm:edge-sample} in \Cref{sec:edge-sample}. 

Thus, we find a proper coloring for both parties' \( \Delta \)-perfect matchings using either \textsf{special} color or color from other party's palette.

In an earlier version (\href{https://arxiv.org/abs/2412.12589v1}{arXiv:2412.12589v1}) of the paper, the protocol of \Cref{lm:edge-sample} was randomized and required  $O(\log^\ast \Delta)$ rounds. We are grateful to Maxime Flin for sharing the idea that led to the improvement to the current deterministic $O(1)$-round protocol.

\paragraph{Coloring deferred subgraph \textsc{DG}:} We now proceed to color Alice's deferred subgraph \textsc{DG} using Bob's palette.

\begin{lemma}
\label{lm:defer}
There exists a protocol that finds a proper coloring of the deferred subgraph \textsc{DG} using \underline{seven} colors from the other party, requiring \(O(n)\) bits of communication and completing in two rounds of communication.
\end{lemma}

\begin{proof}
As established in \Cref{lm:edge-defer}, the deferred subgraph \textsc{DG} has a maximum degree of \(2\). In Alice's \textsc{DG}, the vertices that were originally of degree \(\Delta - 1\) have \emph{at most one} edge in Bob's graph. Additionally, each vertex is incident to \emph{at most one} edge in the \(\Delta\)-perfect matching that may use a color from Bob's palette.

We consider the first seven colors in Bob's palette, denoted as \(\{1, 2, 3, 4, 5, 6, 7\}\). We assume that \(\Delta \geq 8\), as a simple deterministic protocol with constant $\Delta$ is described in \Cref{lem:constant-degree-edge-col}. After coloring Bob's remaining subgraph using his palette, he will send Alice the list of available colors from this set for each vertex, utilizing \(O(n)\) bits of communication.

Since each vertex in Alice's deferred subgraph is incident to at most one edge in Bob's graph and at most one edge in the \(\Delta\)-perfect matching, there are at least five available colors for each vertex. This is sufficient for Alice to perform a greedy coloring of the deferred subgraph.

Consider an edge \(\{u, v\}\) in \textsc{DG}. Since both vertices \(u\) and \(v\) have at least five available colors, they share at least three colors in common. Given that the maximum degree of \textsc{DG} is \(2\), there exists at least one available color that can be used to color the edge \(\{u, v\}\), regardless of the colors of other edges incident to \(u\) and \(v\).

Thus, Alice can successfully color the deferred subgraph \textsc{DG}.
\end{proof}

\mainedgeupper*

\begin{proof}
We now demonstrate the correctness and communication complexity of the protocol.

\paragraph{Correctness of the protocol:} For the case when \(\Delta < 7\), we apply the deterministic protocol outlined in \Cref{lem:constant-degree-edge-col}. Otherwise, for \(\Delta \geq 8\), Alice and Bob use \Cref{alg:EdgeColoringProtocol}. Upon identifying the deferred subgraph and the \(\Delta\)-perfect matching, the remaining subgraph has a maximum degree of \(\Delta - 1\). Notably, vertices with degree \(\Delta - 1\) form an independent set within this subgraph. According to \Cref{thm:independent}, there exists a proper edge-coloring utilizing \(\Delta - 1\) colors, which Alice and Bob can find locally. As both parties use their distinct palettes of size \(\Delta - 1\), the coloring is valid at this stage.

Next, Alice and Bob will proceed to color their respective \(\Delta\)-perfect matchings. They are aware of the \emph{adjacency} of edges in their matchings by sending an array of size $n$ indicating the vertices covered by their matchings. If an edge in one party's matching is not adjacent to any edge in the other's matching, it will be colored with a \textsf{special} color. Conversely, if two edges in the matchings are adjacent, we identify a vertex incident to these edges, which must have a degree of at most \(\Delta / 2\) in at least one of the parties. By using \Cref{lm:edge-sample}, that party can find an available color from the other party's palette to color their edge. 

The coloring remains valid, as the \textsf{special} color is exclusively utilized within the matchings. Furthermore, when two adjacent edges in the matchings are found, at least one party, for instance, Alice, will sample from Bob's palette to color her edge in the matching. This approach is valid because one endpoint of the edge corresponds to a vertex of degree \(\Delta\) in Alice's graph, which is adjacent solely to edges colored with Alice's palette. The other endpoint is a vertex from which Alice and Bob have sampled for an available color from Bob's palette.

Finally, Alice and Bob color their respective deferred subgraphs according to \Cref{lm:edge-defer}, the correctness of which has been established in the proof. Thus, we conclude that \Cref{alg:EdgeColoringProtocol} successfully generates a proper \((2 \Delta - 1)\)-edge-coloring for both Alice and Bob.

%\gopinath{(Though it is obvious, may be a very short description  why the algorithm produces a correct output would be great.} %\yijun{I feel that the structure of the proof is a bit messy here. My feeling is that the analysis of the algorithm mainly depends on \Cref{lm:matching,lm:edge-sample,lm:defer}. They were mentioned in the algorithm description and should also be mentioned here. On the other hand, I found  to be quite strange and I don't see why we need this (it does not seem to add more information compared with the lemma before). I suggested removing this because it is not mentioned in the algorithm description.} 

\paragraph{Communication complexity of the protocol:} Regarding communication complexity, all steps in the protocol are deterministic and require \(O(n)\) bits across a constant number of communication rounds. Therefore, we conclude that the overall communication complexity of \Cref{alg:EdgeColoringProtocol} aligns with the analysis provided.
\end{proof}

\section{Lower bound for \texorpdfstring{$(2\Delta-1)$}{(2Delta-1)}-edge coloring}\label{sec:lowerbound}

In this section, we demonstrate that every constant-error protocol for the $(2\Delta - 1)$-edge coloring problem requires $\Omega(n)$ bits of communication in the worst case, thereby proving \Cref{them:lower}. In \Cref{sec:zero-comm-protocol}, we show that any constant-error protocol for the $(2\Delta - 1)$-edge coloring problem that communicates $o(n)$ bits implies the existence of a zero-communication protocol with a success probability of $2^{-o(n)}$. In \Cref{sec:zero-comm-game}, we construct a zero-communication game, named as \emph{zero communication edge coloring} game (\zec), that can be interpreted as solving the $(2\Delta - 1)$-edge coloring problem in a constant-sized graph without communication. Furthermore, using the parallel repetition theorem \cite{raz1998parallel, Holenstein_2009}, we show that any zero-communication protocol that solves $n$ independent instances of \zec admits an upper bound of $2^{-\Omega(n)}$ on the success probability. In \Cref{sec:final-lower-proof}, we present the desired lower bound proof for the $(2\Delta-1)$-edge coloring problem, where the hard graph instance used in the proof is derived from $n$ independent instances of  \zec. Finally, in \Cref{sec:wstream}, we show that the $(2\Delta-1)$-edge coloring problem admits a space lower bound of $\Omega(n)$ bits in the $W$-streaming model.
% \Cref{sec:zero-comm-protocol} and \Cref{sec:zero-comm-game}.

\subsection{Communication guessing protocol}
\label{sec:zero-comm-protocol}
Though we state the following lemma about zero-communication protocol for the $(2\Delta - 1)$-edge coloring problem, the lemma is quite general and is applicable to other problems.
%\yijun{While I said that the following lemma is folklore, I don't think we need to explicitly say that, unless we have an evidence, e.g., some papers that already used this.}\gopinath{Yes, I agree.}
\begin{lemma}
\label{lemma:zero-comm-protocol}
    Assume that there exists a constant-error randomized protocol $\mathcal{P}$ that solves the $(2\Delta - 1)$-edge coloring problem using $o(n)$ bits worst-case. Then, there exists a zero-communication protocol without public randomness, which solves the $(2\Delta - 1)$-edge coloring problem with a success probability of $2^{-o(n)}$.
\end{lemma}
%\gopinath{This lemma is nothing  specific about $(2\Delta-1)$-edge coloring. Therefore, may be we put it in preliminaries.}\yijun{I am fine with having the lemma here. While this lemma is quite general, it is easier to understand if we have a concrete problem in mind.}\farrel{Maybe add remark that this can be applied for any problem and not just coloring?}\yijun{I would not try to highlight/emphasize the lemma, as most likely this is known/folklore.}
\begin{proof}
First, by \cite{newman1991private}, there exists a protocol $\mathcal{P}'$ that solves the $(2\Delta - 1)$-edge coloring problem using $o(n)$ bits in the worst case, without requiring public randomness. In protocol $\mathcal{P}'$, Alice and Bob exchange a communication pattern $\mathcal{M} \in \{0,1\}^{o(n)}$ to coordinate their actions and find a valid edge coloring.

Now, consider a new protocol, $\mathcal{P}_0$, where Alice and Bob attempt to simulate $\mathcal{P}'$ without exchanging any information or using public randomness. Specifically, in $\mathcal{P}_0$, both Alice and Bob independently and uniformly guess the communication pattern $\mathcal{M}$ that would have been exchanged in $\mathcal{P}'$. Let $\mathcal{M}_A$ and $\mathcal{M}_B$ denote the guesses made by Alice and Bob, respectively, both drawn from $\{0,1\}^{o(n)}$. Each player then uses their respective guesses, $\mathcal{M}_A$ and $\mathcal{M}_B$, to simulate the actions they would take in $\mathcal{P}'$ and produce an output based on those actions.

With a probability of $2^{-o(n)}$, Alice’s and Bob’s guesses will match the actual communication pattern, i.e., $\mathcal{M}_A = \mathcal{M}_B = \mathcal{M}$. In this case, protocol $\mathcal{P}_0$ behaves exactly like $\mathcal{P}'$, successfully producing a valid edge coloring. Thus, $\mathcal{P}_0$ serves as a zero-communication protocol that solves the edge coloring problem with a success probability of at least $2^{-o(n)}$.
\end{proof}

\subsection{Zero-communication edge coloring (\zec) game}
\label{sec:zero-comm-game}

The \zec game is a cooperative game between two players, Alice and Bob, in which they must solve the $(2\Delta-1)$-edge coloring problem on an underlying graph $\mathbf{G}$ (chosen from a probability distribution) with $\Delta=2$.

\begin{description}
    \item[The underlying graph $\mathbf{G}$:] Alice and Bob are given a fixed set of nine vertices, $V = \{v_A, v_B, v_1, v_2, v_3, v_4, v_5, v_6, v_7\}$, which is known to both of them. A referee selects a set of edges $\bea$, for the input of Alice, chosen uniformly at random from 
    \[
    \mathcal{S}_A = \left\{ \{ \{v_A, v_i\}, \{v_A, v_j\} \} \ | \ 1 \leq i < j \leq 7 \right\}.
    \]
    Similarly, independent from $\bea$, the referee selects a set of edges $\beb$, for the input of Bob, chosen uniformly at random from
    \[
    \mathcal{S}_B = \left\{ \{ \{v_B, v_i\}, \{v_B, v_j\} \} \ | \ 1 \leq i < j \leq 7 \right\}.
    \]
    The resulting graph is $\mathbf{G} = (V, \bea \cup \beb)$, which is 3-edge colorable since the maximum degree $\Delta = 2$. Both Alice and Bob know that the set of acceptable colors for the edges of $\mathbf{G}$ is $\{c_1, c_2, c_3\}$.

    \item[Objective and winning condition:] Without any communication or access to public randomness, Alice and Bob must each assign a color from $\{c_1, c_2, c_3\}$ to the edges they receive. They win jointly if their combined assignments result in a valid 3-edge coloring of the graph $\mathbf{G} = (V, \bea \cup \beb)$. 
\end{description}

   Note that the outputs of Alice and Bob are independent because each player’s output is determined solely by their respective input and private random bits. Consequently, any algorithm employed by Alice and Bob to succeed in the \zec game (i.e., to solve the $(2\Delta - 1)$-edge coloring problem on $\mathbf{G}$) can be viewed as a combination of the individual strategies used by the two players.

    %Furthermore, the \zec game can be viewed as finding $(2\Delta - 1)$-edge coloring of a $O(1)$-vertices graph where $\Delta = 2$ with zero communication and no access to public randomness.%\gopinath{May be we specify $\Delta=2$?}

\begin{lemma}
\label{lemma:game-prob}
    For every (randomized) algorithm $\mathcal{A}$ used by Alice and Bob, the probability for them to jointly win the 
    \zec game is bounded by some constant $c<1$. More specifically,
    \[\Pr[\mathcal{A} \text{ gives a proper 3-edge coloring of $\mathbf{G}$}]\le \frac{11024}{11025}\]
\end{lemma}
%\gopinath{I guess the claim is for a deterministic lower bound when the input is coming from a hard distribution?}\yijun{As discussed in the Monday meeting, I think we need a randomized lower bound for a fixed distribution for the lower bound to work.}

\begin{proof}
    Consider any (randomized) strategy $\mathcal{A}$ used by the players. Based on the behavior of $\mathcal{A}$,  consider the following labeling of the vertices in $\{v_1,\ldots,v_7\}$ w.r.t. Alice and Bob.%Denote $\{c_1, c_2, c_3\}$ as the available palette for the 3-coloring of the edges.

    For each vertex $v_i \in \{v_1, \ldots, v_7\}$, the label of $v_i$ w.r.t. Alice is denoted by $\labela (v_i)\subseteq \{c_1, c_2, c_3\}$ and it contains only those $c_j$ such that there exists a possible  input set of edges $S_A \in \mathcal{S}_A$ for Alice with
    
    \[
        \Pr[\text{Alice colors the edge }\{v_A, v_i\} \ \text{with }c_j \mid \bea =  S_A]\ge \frac{1}{5}.
    \]

    Similarly, we construct the label of $v_i$ w.r.t. Bob is denoted by $\labelb(v_i) \subseteq \{c_1, c_2, c_3\}$ such that it contain only those $c_j$ there exists  a possible input set of edge $S_B \in \mathcal{S}_B$ for Bob with

    \[
        \Pr[\text{Bob colors the edge }\{v_B, v_i\} \ \text{with }c_j \mid \beb = S_B]\ge \frac{1}{5}.
    \]

    Next, define set $V_A$ containing all $v_i \in \{v_1, v_2, \dots, v_7 \}$ such that $|\labela(v_i)| = 1$. Analogously, define $V_B$ as the set containing all $v_i \in \{v_1, v_2, \dots, v_7 \}$ with $|\labelb(v_i)| = 1$.
    
    Next, we categorize the analysis into two distinct cases, depending on the cardinalities of \(V_A\) and \(V_B\).
    
    % when l_A(v_1) = l_A(v_2) = {c_j}
    \paragraph{Case 1:} If $|V_A| \ge 4$ or $|V_B| \ge 4$, assume without loss of generality that $|V_A| \ge 4$. By pigeonhole principle, there exists distinct $v_{i_1}, v_{i_2} \in V_A$, and a color $c_j \in \{c_1,c_2,c_3\}$ such that
    \[
        \labela(v_{i_1})=\labela(v_{i_2})=\{c_j\}.
    \]

Let us   define $\mathcal{E}_1$ and $\mathcal{E}_2$ as the events that $c_j$ is the color assigned to $(v_A, v_{i_1})$ and $(v_A, v_{i_2})$, respectively. Define also $\mathcal{E}$ as the event that the input set edges given to Alice is $\{\{v_A, v_{i_1}\}, \{v_A, v_{i_2}\} \}$, i.e., $\bea$ is $\{\{v_A, v_{i_1}\}, \{v_A, v_{i_2}\} \}$. Then,
    \[
        \Pr[\mathcal{E}_1 \mid \mathcal{E}] = 1 - \sum_{j' \neq j} \Pr[\{v_A, v_{i_1}\} \text{ is colored with }c_{j'} \mid \mathcal{E}] \ge 1 - \sum_{j' \neq j} \frac{1}{5} = \frac{3}{5}.
    \]
    Similarly, we can get $\Pr[\mathcal{E}_2 \mid \mathcal{E}] \ge \frac{3}{5}$. Hence,
    \[
        \Pr[\mathcal{E}_1 \cap \mathcal{E}_2 \mid \mathcal{E}] = \Pr[\mathcal{E}_1 \mid \mathcal{E}] + \Pr[\mathcal{E}_2 \mid \mathcal{E}] - \Pr[\mathcal{E}_1 \cup \mathcal{E}_2 \mid \mathcal{E}] \ge \frac{3}{5} + \frac{3}{5} - 1 = \frac{1}{5}.
    \]
Since, $\bea$ is chosen uniformly at random from $\mathcal{S}_A$, 

$$\Pr[\mathcal{E}] = \frac{1}{|\mathcal{S}_A|} = \frac{1}{\binom{7}{2}} = \frac{1}{21}.$$

Observe that,   if all of $\mathcal{E}_1, \mathcal{E}_2, \mathcal{E}$  happen, then  the two incident edges $\{v_A, v_{i_1}\}$ and $\{v_A, v_{i_2}\}$ are colored with the same color $c_j$. This will imply that the resulting edge coloring of $\mathbf{G}$ is not proper. Therefore,
    \begin{align*}
        \Pr[\text{The strategy  gives aproper $3$-coloring of } \mathbf{G}] &\le 1 - \Pr[\mathcal{E}_1 \cap \mathcal{E}_2 \cap \mathcal{E}] \\
        &= 1 - \Pr[\mathcal{E}] \cdot \Pr[\mathcal{E}_1 \cap \mathcal{E}_2 \ | \ \mathcal{E}] \\
        &\le 1 - \frac{1}{21} \cdot \frac{1}{5}\\
        &= \frac{104}{105}.
    \end{align*}

%    Thus, proving the lemma in this case.

    \paragraph{Case 2:} Otherwise, we have $|V_A| \le 3$ and $|V_B| \le 3$. Then, there exists $v_i \in \{v_1, \ldots, v_7\}$ such that $v_i \notin V_A$ and $v_i \notin V_B$. We have $|\labela(v_i)| \ge 2$ and $|\labelb(v_i)| \ge 2$, and hence by pigeonhole principle, exists common color $c_{j'} \in \labela(v_i) \cap \labelb(v_i)$.

    By definition of $\labela(v_i)$, there exists possible input set of  edges $T_A \in \mathcal{S}_A$ for Alice such that $\Pr[\{v_A,v_i\} \text{ is colored with } c_{j'} \mid \bea = T_A] \ge \frac{1}{5}$. Similarly, there also exists a possible input set of edges $T_B \in \mathcal{S}_B$ for Bob such that $\Pr[\{v_B,v_i\} \text{ is colored with }c_{j'} \mid \beb = T_B] \ge \frac{1}{5}$.

    Let us define  $\mathcal{E}_A$ and $\mathcal{E}_B$ as the events that Alice is given input $T_A$ and Bob is given input $T_B$, respectively. Also, define $\mathcal{E}'_A$ and $\mathcal{E}'_B$ as the events that $c_{j'}$ is the color of edge $\{v_A, v_i\}$ and $\{v_B, v_i\}$ respectively.  Note that $\Pr[\mathcal{E}'_A \mid \mathcal{E}_A] \ge \frac{1}{5}$ and $\Pr[\mathcal{E}'_B \mid \mathcal{E}_B] \ge \frac{1}{5}$. We can then get:
    \begin{align}
        \Pr[\mathcal{E}_A \cap \mathcal{E}_B \cap \mathcal{E}'_A \cap \mathcal{E}'_B] &= \Pr[\mathcal{E}_A \cap \mathcal{E}'_A] \cdot \Pr[\mathcal{E}_B \cap \mathcal{E}'_B] &&\tag{1}\label{eq1} \\
        &= \Pr[\mathcal{E}_A] \cdot \Pr[\mathcal{E}'_A \mid \mathcal{E}_A] \cdot \Pr[\mathcal{E}_B] \cdot \Pr[\mathcal{E}'_B \mid \mathcal{E}_B] &&\tag{2}\label{eq2}\\
        &\ge \frac{1}{21} \cdot \frac{1}{5} \cdot \frac{1}{21} \cdot \frac{1}{5} &&\tag{3}\label{eq3} \\
        &= \frac{1}{11025} &&\notag
    \end{align}
    Where \Cref{eq1} comes from the independence of $\mathcal{E}_A \cap \mathcal{E}'_A$ and $\mathcal{E}_B\cap  \mathcal{E}'_B$ as the input of Alice and Bob are independent of each other and the output of Alice and Bob are independent of each other; \Cref{eq2} comes from Bayes theorem; \Cref{eq3} comes from the fact that $\Pr[\mathcal{E}_A] = \Pr[\mathcal{E}_B] = \frac{1}{\binom{7}{2}} = \frac{1}{21}$, $\Pr[\mathcal{E}'_A \mid \mathcal{E}_A] \ge \frac{1}{5}$, $\Pr[\mathcal{E}'_B \mid \mathcal{E}_B] \ge \frac{1}{5}$.
    %\yijun{Can write ${7 \choose 2} = 21$ so it becomes clearer.}
    %Where equation (1) comes from the independence of $\mathcal{E}_A$ and $\mathcal{E}_B$; equation (2) comes from the independence of $\mathcal{E}_1$ and $\mathcal{E}_B$ as Alice won't have any knowledge about Bob's edges; equation (3) comes from the fact that Bob's output only depends on the given input edges \farrel{Can check this statement? I'm not sure if the reasoning for (3) is correct.}\yijun{Looks correct. I would simplify the calculation until (3) into just 1 or 2 lines: first use independence to write $\Pr[\mathcal{E}_A \cap \mathcal{E}_B \cap \mathcal{E}_1 \cap \mathcal{E}_2]=\Pr[\mathcal{E}_A \cap \mathcal{E}_1] \cdot \Pr[\mathcal{E}_B \cap \mathcal{E}_2]$, and then jump to equation (3)}; inequality (4) comes from the fact that $\Pr[\mathcal{E}_A] = \Pr[\mathcal{E}_B] = \frac{1}{21}$, $\Pr[\mathcal{E}_1 \mid \mathcal{E}_A] \ge \frac{1}{5}$, $\Pr[\mathcal{E}_2 \mid \mathcal{E}_B] \ge \frac{1}{5}$.
    
    Observe that when $\mathcal{E}_A, \mathcal{E}_B, \mathcal{E}_1, \mathcal{E}_2$ are all happening, the resulting edge coloring of $\mathbf{G}$ is not proper as two incident edges $\{v_A, v_i\}$ and $\{v_B, v_i\}$ are assigned with the same color $c_{j'}$. Hence,
    \[
        \Pr[\text{The strategy gives a proper 3-coloring}] \le 1 - \Pr[\mathcal{E}_A \cap \mathcal{E}_B \cap \mathcal{E}_1 \cap \mathcal{E}_2] \le 1 - \frac{1}{11025} = \frac{11024}{11025}
    \]
Combining the analysis of Case 1 and 2, we are done with the proof of \Cref{lemma:game-prob}.
    %Once again, proving the lemma.
\end{proof}
Now, we show that $2^{-\Omega(n)}$ is the upper bound on the probability with which Alice and Bob can solve $n$ independent instances of \zec game. We rely on the parallel repetition theorem of Raz~\cite{Holenstein_2009,raz1998parallel}.

\begin{proposition}[{Parallel repetition theorem~\cite{Holenstein_2009,raz1998parallel}}]
\label{thm:parallel-game}
    Suppose that there is a game where a referee selects a pair $(x,y)$ from a known joint distribution $P_{XY}$. The referee sends $x$ to Alice and $y$ to Bob, who then provide outputs $a$ and $b$ respectively, without communicating with each other. Alice and Bob win the game if a known condition $Q(x,y,a,b)$ is satisfied. If the highest probability for Alice and Bob to win this game is $v<1$, then the probability for them to win all $n$ independent instances of this game is at most $\left(1 - \frac{(1-v)^3}{6000}\right)^{\frac{n}{\log s}}$, where $s$ is the maximum number of possible outputs $(a,b)$ of the original game.
\end{proposition}

\begin{lemma}
\label{lemma:low-success-prob}
Any randomized algorithm that solves   $n$ independent instances of the \zec game admits a success probability upper bounded by $2^{-\Omega(n)}$.%\yijun{I would call this a lemma. A corollary is something of independent interest.}
\end{lemma}%\gopinath{Giving a name to the edge-coloring game will help this statement to be more precise.}
\begin{proof}
    The proof follows from \Cref{lemma:game-prob} and \Cref{thm:parallel-game}.
\end{proof}

\subsection{Final lower bound and the proof of \Cref{them:lower}}
\label{sec:final-lower-proof}

\edgecoloringlower*

\begin{proof}
Assume, for the sake of contradiction, that there exists a constant-error randomized protocol that solves the $(2\Delta - 1)$-edge coloring problem using $o(n)$ bits in the worst case. By \Cref{lemma:zero-comm-protocol}, this would imply the existence of a zero-communication protocol, without access to public randomness, that solves the $(2\Delta - 1)$-edge coloring problem with a success probability of $2^{-o(n)}$.

Next, consider the problem of solving $n$ independent instances of the \zec game. This is equivalent to finding a $(2\Delta - 1)$-edge coloring of some specific kind of graphs with $O(n)$ vertices, where the edges are partitioned between two parties, without any communication or access to public randomness. By \Cref{lemma:low-success-prob}, the probability that Alice and Bob successfully find a $(2\Delta - 1)$-edge coloring for such a graph is at most $2^{-\Omega(n)}$, which is worse than $2^{-o(n)}$. This leads to a contradiction.
%Thus, \Cref{them:lower} is proved.%\yijun{Try to write the proof of the main theorem into a proof environment.}
\end{proof}

\subsection{Space lower bound in the \texorpdfstring{$W$}{W}-streaming model}\label{sec:wstream}
%\gopinath{I will edit the section tomorrow.}
Consider reducing two-party communication complexity to the standard streaming model, where output size also counts toward space usage. If a problem $P$ has an $r$-pass, $s$-space streaming algorithm, then Alice and Bob can simulate it in the two-party communication model using $rs$ bits. Alice runs the algorithm on her input and sends the memory state of the streaming algorithm to Bob, who continues with his input and sends the updated memory state back. Repeating this $r$ times yields the desired protocol. Therefore, from \Cref{them:lower} and \Cref{cor-lb-main}, we have an $\Omega(n)$-bit space lower bound for the $(2\Delta-1)$-edge coloring problem in the standard streaming model for a constant number of passes. However, an $\Omega(n)$-bit lower bound is trivial in the context of the $(2\Delta-1)$-edge coloring in the standard streaming model, since the output size is as large as the input size. Therefore, as already mentioned in \Cref{sec:contribution}, the $W$-streaming model is a more natural model for the $(2\Delta-1)$-edge coloring problem. Recall that, in the $W$-streaming model, the output is also allowed to be reported in the streaming fashion and the space complexity of an algorithm is just the internal storage of the algorithm.

Observe that the general reduction of the two-party communication model to standard streaming does not apply to the reduction of the $(2\Delta-1)$-coloring problem in the two-party model to the same problem in the $W$-streaming model. This is because, in our setup of the $(2\Delta-1)$-edge coloring problem, Alice and Bob need to report the colors of their respective edges. However, a $W$-streaming algorithm may report the color of an edge belonging to Alice while processing Bob's edges.  

To address this issue, we consider the following relaxed version of the $(2\Delta-1)$-edge coloring problem in the two-party communication model:  

\paragraph{Weaker-$(2\Delta - 1)$-edge coloring:}  The input setup for the weaker-$(2\Delta - 1)$-edge coloring problem in the two-party model is the same as that of the $(2\Delta - 1)$-edge coloring problem, i.e., the edges of the input graph are partitioned between Alice and Bob. However, unlike in the $(2\Delta - 1)$-edge coloring problem, Alice and Bob are not required to report the colors of their respective edges. Instead, each party is allowed to report the color of any edges, as long as the color of every edge is output by at least one of them.

Observe that we can reduce the weaker-$(2\Delta-1)$-edge coloring problem in the two-party communication model to the $(2\Delta-1)$-edge coloring problem in the $W$-streaming model. Therefore, let us first discuss how we can modify the proof of \Cref{them:lower} to establish a communication complexity lower bound of $\Omega(n)$ bits for the weaker-$(2\Delta-1)$-edge coloring problem.

 %Hence, showing that the communication complexity of  weaker-$(2\Delta-1)$-edge coloring is $\Omega(n)$ bits is enough to get \Cref{coro:w}.

 \paragraph{Modified \zecnew game:} Analogous to the \zec game used to prove \Cref{them:lower}, we consider a modified version of \zec.  Previously, in the \zec game, Alice receives the edges $\{\{v_A, v_i\}, \{v_A, v_j\}\}$, where $i$ and $j$ are sampled uniformly from $\{1, 2, \dots, 7\}$. In the modified \zecnew game, Alice instead receives the edges $\{\{v_A^\ast, v_i\}, \{v_A^\ast, v_j\}\}$, where the vertex $v_A^\ast$ is sampled uniformly from the set of vertices $\{v_{A_1}, v_{A_2}, \dots, v_{A_{33075}}\}$, and $i, j$ are sampled uniformly from $\{1, 2, \dots, 7\}$. Similarly, Bob receives the edges $\{\{v_B^\ast, v_i\}, \{v_B^\ast, v_j\}\}$, where $v_B^\ast$ is sampled uniformly from $\{v_{B_1}, v_{B_2}, \dots, v_{B_{33075}}\}$, and $i, j$ are sampled uniformly from $\{1, 2, \dots, 7\}$.
We say that Alice and Bob win the game if at least one of the following winning conditions are met:
\begin{enumerate}
    \item Alice and Bob color their respective edges using three colors, and the combined result forms a valid 3-edge coloring;
\item Alice correctly guesses Bob's vertex $v_B^\ast$;
\item Bob correctly guesses Alice's vertex $v_A^\ast$.
\end{enumerate}

Using \Cref{lemma:game-prob}, the probability that a strategy $\mathcal{A}$ wins this modified \zecnew game is bounded by
\begin{align*}
&\Pr[\text{A proper 3-edge coloring is given}]\\
&+ \Pr[\text{Alice correctly guesses } v_B^\ast] + \Pr[\text{Bob correctly guesses } v_A^\ast] \\
    & \leq \frac{11024}{11025} + \frac{1}{33075} + \frac{1}{33075} = \frac{33074}{33075}. 
\end{align*}

Suppose there exists a constant-error randomized protocol that solves the weaker-$(2\Delta - 1)$-edge coloring using $o(n)$ bits in the worst case. Using the proof of \Cref{lemma:zero-comm-protocol}, this implies the existence of a zero-communication protocol $\mathcal{P}_0$ that solves the weaker-$(2\Delta - 1)$-edge coloring problem with a success probability of $2^{-o(n)}$.  

Now, consider solving the weaker-$(2\Delta - 1)$-edge coloring problem on a graph $G$, where $G$ is constructed as the union of the underlying graphs of $n$ independent \zecnew game instances, $(\mathcal{G}_1, \mathcal{G}_2, \dots, \mathcal{G}_n)$. For each \zecnew game instance $\mathcal{G}_i$:

\begin{itemize}
    \item If Alice and Bob each outputs the color for their own edges in the underlying graph of $\mathcal{G}_i$ and these colors form a valid edge coloring, then game $\mathcal{G}_i$ is won.
    \item If Alice outputs some of Bob's edges in the underlying graph of $\mathcal{G}_i$ and the overall coloring is valid, then Alice must know Bob's vertex $v_B^\ast$ in the underlying graph $\mathcal{G}_i$, which counts as a win.
    \item Similarly if Bob outputs some of Alice's edges in the underlying graph of $\mathcal{G}_i$ and the overall coloring is valid, then Bob must know Alice's vertex $v_A^\ast$, and $\mathcal{G}_i$ is won.
\end{itemize}

Thus, protocol $\mathcal{P}_0$ would allow Alice and Bob to win all $n$ independent instances of the \zecnew game with probability $2^{-o(n)}$. On the other hand, by \Cref{thm:parallel-game}, the probability for both parties to win all $n$ independent instances of the \zecnew game is bounded by $2^{-\Omega(n)}$, which leads to a contradiction. Hence, $\Omega(n)$ bits are required to solve the weaker $(2\Delta - 1)$-edge coloring problem.

Hence, we have the following theorem analogous to \Cref{them:lower}:

\begin{theorem}
      Any randomized protocol that solves the weaker-$(2\Delta - 1)$-edge coloring problem with probability $1/2$ requires $\Omega(n)$ bits of communication in the worst case.
\end{theorem}

In a similar way to how we argued that \Cref{them:lower} implies \Cref{cor-lb-main} in \Cref{sec:contribution}, we can also argue that the expected communication complexity of the weaker-$(2\Delta - 1)$-edge coloring problem is $\Omega(n)$ bits. Now, considering the reduction of the weaker-$(2\Delta - 1)$-edge coloring problem to the same problem in the $W$-streaming model, we obtain the desired result in the $W$-streaming model.

\wstream*

%\farrel{Need to explain that this implies $\Omega(n)$ lower bound in W-Stream? I'm not sure how to do this.}

\section{Conclusions and open problems}\label{sect:conclusions}

In this work, we presented \emph{round-efficient} protocols for the \((\Delta + 1)\)-vertex coloring problem and the \((2\Delta - 1)\)-edge coloring problem  in the two-party communication model with \emph{optimal} communication complexity $O(n)$. %Moreover, we demonstrated the optimality of the communication cost for the edge coloring protocol by proving the communication lower bound in this setting. 
We now discuss several remaining open questions.

\paragraph{Round complexity:} Is it possible to design protocols for the \((\Delta + 1)\)-vertex coloring problem and the \((2\Delta - 1)\)-edge coloring problem that use \(O(n)\) bits of communication while maintaining an \emph{asymptotically minimum} number of rounds? To the best of our knowledge, there is no known round complexity lower bound for these problems. If it is impossible to achieve both optimal round and communication complexities simultaneously, what is the best round-communication tradeoff?  %Although our protocols are round-efficient, it remains open whether they achieve the optimal round complexity while keeping the communication cost within \(O(n)\) bits.

\paragraph{Deterministic and high probability protocols:} For the \((\Delta + 1)\)-vertex coloring problem, we focus on designing \emph{Las Vegas} protocols and analyzing their \emph{expected} round complexity. Very little is known about protocols that are deterministic or succeed \emph{with high probability}, which means that the success probability is $1 - 1/\poly(n)$.
It is an interesting future research direction to explore the round-communication tradeoff in these settings. What is the optimal communication complexity for the \((\Delta + 1)\)-vertex coloring problem in these settings?

For the \((\Delta + 1)\)-vertex coloring problem, \citeauthor{10.1145/3662158.3662796}~\cite{10.1145/3662158.3662796} showed the existence of a deterministic protocol with communication complexity \(O(n \log \Delta)\) and a randomized protocol with \(O(n \log^* \Delta)\) communication complexity that succeeds with high probability. However, these algorithms are not round-efficient. For the round-communication tradeoff, \citeauthor{assadi2023coloring}~\cite[Corollary 3.11 of arXiv:2212.10641v1]{assadi2023coloring} presented a deterministic protocol using \(O(n \log^4 n)\) bits of communication and requiring only \(O(\log \Delta \log \log \Delta)\) rounds.

%A natural open question is whether the communication complexity of \((\Delta + 1)\)-vertex coloring can be reduced to \(O(n)\) in both deterministic and randomized high-probability settings, while maintaining the efficient round complexity.

%For edge coloring, the \emph{only} step requiring randomness in \Cref{alg:EdgeColoringProtocol} is the sampling of a color for the vertices in both \(\Delta\)-perfect matchings. By Hall's marriage theorem, there exist two disjoint \(\Delta\)-perfect matchings corresponding to the two parties, and these matchings require only \(O(n)\) bits of space to store. There is no known deterministic \((2\Delta - 1)\)-edge coloring protocol that is both round and communication efficient. %An interesting open question is to design a deterministic \((2\Delta - 1)\)-edge coloring protocol.

\paragraph{Fewer colors:} A natural research direction is to reduce the number of colors. \emph{Brooks' theorem} shows the existence of a $\Delta$-vertex coloring for any graph that is not a complete graph or an odd-length cycle. 
\citeauthor{10.1145/3662158.3662796} \cite[Problem 5]{10.1145/3662158.3662796} posed an open problem to determine the optimal communication complexity to compute a $\Delta$-vertex coloring, which remains unsolved. 
For edge coloring, Vizing's theorem shows the existence of a $(\Delta+1)$-edge coloring for any graph. What is the optimal communication complexity to compute a $(\Delta+1)$-edge coloring?

\section*{Acknowledgments}
In an earlier version (\href{https://arxiv.org/abs/2412.12589v1}{arXiv:2412.12589v1}) of the paper, our $(2\Delta-1)$-edge coloring algorithm was randomized and required  $O(\log^\ast \Delta)$ rounds. We are grateful to Maxime Flin for sharing the idea that led to the improvement to the current deterministic $O(1)$-round protocol.

\printbibliography

\appendix

\section{Algorithm for sampling an available color: Proof of  \Cref{lm:ksi}}\label{app:col-samp}

\citet{10.1145/3662158.3662796} solved the problem of finding an arbitrary {available} color by translating it into a specific kind of set-intersection problem, as follows:

\begin{problem}[$k$-\textsc{Slack}-\textsc{Int}]
    \label{prob:ksi}
    Given two sets \( X \) and \( Y \) with \( |X| + |Y| \leq m - k \) for some \( k \geq 1 \), and a positive integer \( m \), the goal is to find an element in the intersection of \( [m] \setminus X \) and \( [m] \setminus Y \).
\end{problem}

In the context of vertex coloring, finding an arbitrary {available} color for vertex \( v \) is translated into solving $k$-\textsc{Slack}-\textsc{Int} (\Cref{prob:ksi}) with \( m = \Delta + 1 \), \( X = A \), and \( Y = B \). If the number of {available} colors is at least $k$, i.e., \( |[\Delta + 1]\setminus (A \cup B)| \geq k \), then the condition \( |X| + |Y| \leq m - k \) is satisfied because the neighborhoods \( N_A(v) \) and \( N_B(v) \) are disjoint.

Flin and Mittal established that a randomized protocol (as described in \Cref{alg:k-Slack-Int}) exists to solve $k$-\textsc{Slack}-\textsc{Int} 
%\yijun{$k$-\textsc{Slack}-\textsc{Int} or \textsc{Slack}-\textsc{Int} or it doesn't matter?}\gopinath{Done!}
with an expected communication cost of \( O(\log^2((m+1)/k)) \) bits and expected round complexity of $O\left(\log ((m+1)/{k})\right))$. A deterministic protocol, which serves as an ingredient for the randomized one, can achieve this with \( O(\log^2(m)) \) bits of communication and $O(\log m)$ rounds by employing binary search to locate elements.

\begin{lemma}[\citet{10.1145/3662158.3662796}]
    \label{lm:deterministic}
    There exists a deterministic protocol that solves $k$-\textsc{Slack}-\textsc{Int} that makes $O(\log^2 m)$ bits of communication and completes in $O(\log m)$ rounds in the worst case.
\end{lemma}

The randomized algorithm for $k$-\textsc{Slack}-\textsc{Int} makes a guess for \( k \) and samples elements based on a probability \( p \) that ensures \( |S \cap X| + |S \cap Y| < |S| \) occurs with significant probability. Under this condition, \( S \) must contain at least one element not in \( X \) or \( Y \). Subsequently, the deterministic protocol is used to identify that element. The pseudocode of \cite{10.1145/3662158.3662796} for $k$-\textsc{Slack}-\textsc{Int} problem is provided in \Cref{alg:k-Slack-Int}.

\begin{algorithm}[ht!]
\caption{An $O(\log^2 \frac{m}{k})$-bit communication protocol for $k$-\textsc{Slack}-\textsc{Int}}
\label{alg:k-Slack-Int}
\KwIn{Alice gets a set $X \subseteq [m]$, Bob gets a set $Y \subseteq [m]$ such that $|X| + |Y| \leq m - k$.}
\KwOut{Any element from $X \cap Y$.}

\For{$\tilde{k} = m, \frac{m}{2}, \ldots, 1$ \tcp{A sequence of exponentially decreasing guesses for $k$}}{
    Alice and Bob choose $S$ by sampling each element of $[m]$ independently with probability $p = \min\left\{1, \frac{150m}{\tilde{k}^2}\right\}$, using public randomness.\;
    \If{$|S \cap X| + |S \cap Y| < |S|$}{
        Run the deterministic binary search protocol to find an element in $S \setminus (X \cup Y)$.\;
    }
}

\end{algorithm}

The result of Flin and Mittal about $k$-\textsc{Slack}-\textsc{Int} is summarized in the following lemma.

\begin{lemma}[\citet{10.1145/3662158.3662796}]
    \label{lm:rand-fm}
There exists a randomized protocol that solves $k$-\textsc{Slack}-\textsc{Int}, which makes an expected $O\left(\log^2 \frac{m+1}{k}\right)$ bits of communication and completes in $O\left(\log \frac{m+1}{k}\right)$ rounds in expectation.
\end{lemma}

By the above lemma and the connection between the problem of finding an arbitrary {available} color and $k$-\textsc{Slack}-\textsc{Int}, we can say the following: If there are at least $k$ available colors for vertex $v$, then Alice and Bob can find an arbitrary available color making \( O\left(\log^2 \frac{\Delta + 1}{k}\right) \) bits of communication in expectation and spending \( O\left(\log \frac{\Delta + 1}{k}\right) \) rounds in expectation. Hence, in the worst case, the communication cost is $O(\log^2 \Delta)$ and the number of rounds spent is $O(\log \Delta)$.

%From now on, we will treat the process of sampling a color from the remaining palette of each vertex as a black box. It is important to note that this protocol can be executed in \emph{parallel} for all vertices. Therefore, if we execute \cref{alg:k-Slack-Int} for all vertices simultaneously, we only need to focus on the worst-case round complexity, which is \( O(\log m) \).\gopinath{Actually, we are not running \Cref{alg:k-Slack-Int}, rather a slight modification (as described in the below paragraph). May be we give a name to the sampling problem like $k$-\textsc{Slack}-\textsc{Samp}.}

% \begin{lemma}
% \label{lm:ksi}
%     There exists a randomized protocol that solves \( k \)-\textsc{Slack}-\textsc{Int} with \( O\left(\log^2 (m/k) \right) \) bits of communication in expectation and \( O(\log m) \) rounds of communication in the worst case.
% \end{lemma}
\begin{proof}[Proof of \Cref{lm:ksi}]
Now, we argue how a minor modification to \Cref{alg:k-Slack-Int} leads to solving $k$-\textsc{Slack}-\textsc{Int} by outputting an element in $[m] \setminus (X \cup Y)$ uniformly at random. Note that neither \Cref{alg:k-Slack-Int} nor the deterministic subroutine favor any specific element in $[m] \setminus (X \cup Y)$ to be the output. Let Alice and Bob (before running \Cref{alg:k-Slack-Int}) draw a random mapping $[m] \rightarrow [m]$ using a source of public randomness. They then apply this mapping to the set of all $m$ elements before inputting the transformed sets into the algorithm \Cref{alg:k-Slack-Int}. This modification ensures that if they run \Cref{alg:k-Slack-Int} on the permuted set and produce the output, then the output is indeed an element in $[m] \setminus (X \cup Y)$ uniformly at random. Now, instead of \Cref{alg:k-Slack-Int}, if we use the modified version above to find an arbitrary available color for vertex $v$, then the output is indeed a color chosen uniformly at random from the set of available colors, which is the desired output if the goal is to sample an available color uniformly at random. Hence, \Cref{lm:ksi} follows.
\end{proof}

\section{A protocol for the \texorpdfstring{$(\mathsf{degree}+1)$}{(degree+1)}-list coloring: Proof of \Cref{lem:d1lc}}\label{app:d1lc}

\begin{proof}[Proof of \Cref{lem:d1lc}]

The desired protocol for the $(\mathsf{degree}+1)$-list coloring problem (\DLC) is as follows.

\begin{description}
    \item[Step 1:] For each vertex $v$, Alice and Bob run $O(\log^2 n)$ independent instances of $\colorsamp$ (as described in \Cref{lm:ksi}) to obtain a set $L(v)$ of $O(\log^2 n)$ colors from $\Psi(v) = \Psi_A(v) \cap \Psi_B(v)$.

    \item[Step 2:] Both Alice and Bob remove the edges $\{u,v\}$ such that $L(u)$ and $L(v)$ are disjoint. Let the resulting graph be $H$.

    \item[Step 3:] If the number of edges in $H$ is $O(n \log^2 n)$, Alice gathers the entire graph $H$ by making $O(n \log^3 n)$ bits of communication. Alice then attempts to find a vertex coloring of the graph $H$, ensuring that each vertex $v$ is assigned a color from $L(v)$. If successful, she sends the colors assigned to all vertices to Bob.

    \item[Step 4:] If either $H$ has more than $O(n \log^2 n)$ edges or Alice fails to find a vertex coloring of $H$, the entire graph $G$ is gathered at Alice. She then solves the \DLC problem on her own and sends the coloring of all vertices to Bob.
\end{description}

From the description of the protocol, it is evident that the output is a valid solution to the input $\DLC$ instance. We now analyze the communication complexity of the protocol. In Step 1, we execute $O(n \log^2 n)$ instances of \colorsamp. By \Cref{lm:ksi}, the communication cost to perform  Step 1 is $O(n \cdot \log^2 n \cdot \log^2 \Delta)$ bits in the worst case. Note that Step 2 does not incur any communication cost. The communication cost to perform Steps 3 and 4 are $O(n \log^3 n)$ and $O(n^2)$ bits, respectively, in the worst case. However, Step 4 is executed only if $H$ contains more than $O(n \log^2 n)$ edges or if Alice fails to find a vertex coloring of $H$ in Step 3. By \Cref{pro:pal}, the probability of executing Step 4 is at most $1/n^c$, for some large constant $c$. Combining these observations, the expected communication cost of the protocol is $O(n \log^2 n \log^2 \Delta+ n\log^3 n)$ bits. Regarding the round complexity, note that in Step 1, we run $O(n \log^2 n)$ instances of $\colorsamp$ in parallel, implying that the number of rounds required in Step 1 is $O(\log \Delta)$, as stated in \Cref{lm:ksi}. Since each of Steps 2, 3, and 4 requires at most $O(1)$ rounds of communication, the total round complexity of the protocol is $O(\log \Delta)$. Putting things together, \Cref{lem:d1lc} follows. 
\end{proof}

\end{document}